\documentclass{article}
\usepackage{color}
\usepackage{amsmath,amssymb,amsthm}
\usepackage{framed,ascmac}
\usepackage[linesnumbered,ruled]{algorithm2e}
\usepackage{multicol,multirow}
\usepackage{cite}
\usepackage{tikz}
\usetikzlibrary{arrows,positioning,automata,calc}
\usepackage{cleveref}

\crefname{theorem}{Theorem}{Theorems}
\crefname{proposition}{Proposition}{Propositions}
\crefname{lemma}{Lemma}{Lemmas}
\crefname{exmp}{Example}{Examples}
\crefname{corollary}{Corollary}{Corollaries}
\crefname{claim}{Claim}{Claims}
\crefname{remark}{Remark}{Remarks}
\crefname{section}{Section}{Sections}
\crefname{definition}{Definition}{Definitions}
\crefname{example}{Example}{Examples}
\crefname{table}{Table}{Tables}
\crefname{appendix}{Appendix}{Appendices}
\crefname{equation}{Equation}{Equations}
\crefname{algorithm}{Algorithm}{Algorithms}
\crefname{subsection}{Subsection}{Subsections}

\newcommand{\argmin}{\mathop{\rm argmin}}

\newcommand{\weight}{\gamma}
\newcommand{\change}{{\rm change}}
\newcommand{\true}{{\rm true}}
\newcommand{\false}{{\rm false}}
\newcommand{\bs}[1]{\boldsymbol{#1}}
\newcommand{\primal}{x}
\newcommand{\pray}{r}
\newcommand{\dual}{y}
\newcommand{\duall}{z}
\newcommand{\dray}{z}
\newcommand{\nontriv}{{\rm nontriv}}
\newcommand{\size}{{\rm size}}
\newcommand{\val}{{\rm VALUE}}
\SetKwInput{KwInput}{Input}         
\SetKwInput{KwOutput}{Output}

\newcommand{\ot}{\leftarrow}

\renewcommand{\mid}{\,:\,}

\makeatletter
\def\iddots{\mathinner{\mkern1mu\raise\p@
		\hbox{.}\mkern2mu\raise4\p@\hbox{.}\mkern2mu
		\raise7\p@\vbox{\kern7\p@\hbox{.}}\mkern1mu}}
\makeatother

\usepackage[top=30truemm,bottom=30truemm,left=25truemm,right=25truemm]{geometry}

\newtheorem{theorem}{Theorem}[section]
\newtheorem{lemma}[theorem]{Lemma}
\newtheorem{proposition}[theorem]{Proposition}

\newtheorem{definition}[theorem]{Definition}
\newtheorem{example}[theorem]{Example}
\newtheorem{remark}[theorem]{Remark}
\newtheorem{claim}[theorem]{Claim}

\title{A Combinatorial Certifying Algorithm\\ for Linear Programming Problems\\ with Gainfree Leontief Substitution Systems}
\author{Kei Kimura and Kazuhisa Makino}
\begin{document}
	\maketitle
	
	\begin{abstract}
		Linear programming (LP) problems with gainfree Leontief substitution systems have been intensively studied in economics and operations research, and include the feasibility problem of a class of Horn systems, which arises in, e.g., polyhedral combinatorics and logic.
		This subclass of LP problems admits a strongly polynomial time algorithm, where 
		devising such an algorithm for general LP problems is one of the major theoretical open questions in mathematical optimization and computer science.
		Recently, much attention has been paid to devising certifying algorithms in software engineering, since those algorithms enable one to confirm the correctness of outputs of programs with simple computations.
		In this paper, we provide the first combinatorial (and strongly polynomial time) certifying algorithm for LP problems with gainfree Leontief substitution systems.
		As a by-product, we answer affirmatively an open question whether the feasibility problem of the class of Horn systems admits a combinatorial certifying algorithm.
	\end{abstract}
	
	\section{Introduction}
	\label{sec:introduction}
	Linear programming (LP) problems have been at the heart of mathematical optimization, and various algorithms have been proposed to solve LP problems such as the simplex method, the ellipsoid method, and the interior point method~\cite{Van20}. 
	Devising a strongly polynomial time algorithm for LP problems is one of the major theoretical open questions in mathematical optimization and computer science. Furthermore, 
	great efforts have been made to construct strongly polynomial time algorithms to solve LP problems with additional properties such as LP problems such that each constraint and the objective function has at most two nonzero coefficients~\cite{Meg83}, combinatorial LP problems~\cite{Tar86}, and LP problems formulating the maximum generalized flow problem~\cite{OlV20}, 
	as they arise in theory and practice.
	In this paper, we focus on LP problems for 
	Leontief substitution systems. 
	A matrix $A$ is called \emph{Leontief} if each column of $A$ has at most one positive element.\footnote{Leontief matrices defined in this paper are sometimes called \emph{pre-Leontief} matrices in the literature.}
	A linear system of the form 
	\begin{align}
		\begin{aligned}
			\label{eq:LP0}
			A\bs{\primal} &= \bs{b}\\
			\bs{\primal} &\ge \bs{0}.
		\end{aligned}    
	\end{align}
	is called a \emph{Leontief substitution system} if $A$ is Leontief and $\bs{b}$ is nonnegative.
	Leontief matrices and systems were first studied in 1950s within the context of input-output
	analysis in economics (for which Wassily Leontief was awarded the Nobel Prize in economics in 1973; see, Leontief~\cite{Leo51} and Dantzig~\cite{Dan55} for example), 
	and have attracted much attention in economics and operations research.
	There exists a line of research on algorithms for LP problems with Leontief substitution systems; an ${\rm O}(m^3n\log n)$ strongly polynomial algorithm for a special case where $A$ has no more than two nonzero elements in any column~\cite{AdC91}, an ${\rm O}(m^2n)$ strongly polynomial algorithm for a special case of gainfree Leontief substitution systems~\cite{JMR92}, and a simplex algorithm~\cite{CGS97}, where $m$ and $n$ respectively denote the number of equations and variables in \eqref{eq:LP0}.
	The gainfree property will be defined in \cref{sec:preliminaries}; it intuitively says that the corresponding network, which will also be defined later, has no \emph{gain} of flow.
	
	We also remark that Leontief substitution systems play an important role in polyhedral combinatorics and logic.
	For example, Horn systems are related to Leontief substitution systems.
	A matrix $A$ is called \emph{Horn} if each row of $A$ has at most one positive element, and a linear system $A\bs{\dual}\le \bs{c}$ with Horn matrix $A$ is called \emph{Horn}.
	Thus, Horn matrices are exactly transposed Leontief matrices, 
	and the feasibility for Horn systems  coincides with that of the dual of LP problems with Leontief substitution systems.
	The feasibility of Horn systems
	was inspired by the Horn Boolean satisfiability (SAT) problem,  a well-studied subclass of SAT in logic and computer science.
	Horn systems have been intensively studied in the literature~\cite{Glo64,CoA72,MaD02} because they have applications in diverse areas such as logic programs, econometrics, program verification, and lattice optimization.
	Subclasses of Horn systems called difference constraint (DC), unit Horn, and unit-positive Horn systems are also extensively investigated, 
	where a matrix $A$ is \emph{difference} if it is a $\{0,\pm 1\}$-matrix having one +1 and one -1 in each row~\cite{For56,Bel58,Moo59,Gol95}, 
	\emph{unit Horn} if it is a Horn $\{0,\pm 1\}$-matrix~\cite{SuW11,ChS13}, and \emph{unit-positive Horn} if it is an integral Horn matrix with the positive elements being one~\cite{UvG88,SuW11}\footnote{Here, matrix $A$ is unit-positive if and only if $A$ is integral and gainfree in \cite{JMR92}, since in \cite{JMR92} the positive element of $A$ is assumed to be one.}.
	We note that unit and unit-positive Horn systems are also called Horn constraint and extended Horn, respectively.
	By definition, difference matrices are unit Horn, and unit Horn matrices are unit-positive.
	All these matrices are transposed gainfree Leontief matrices, which will be discussed in the next section.
	The feasibility problem is combinatorially solvable in ${\rm O}(mn)$ for DC systems~\cite{For56,Bel58} and 
	${\rm O}(mn^2)$ for unit and unit-positive Horn systems~\cite{ChS13}, where $m$ and $n$ respectively denote the number of inequalities and variables in the system.
	We remark that the feasibility coincides with the integer feasibility for all such Horn systems, where the \emph{integer feasibility} is to ask the existence of an integer vector satisfying a given system. 
	However, this is not true for general Horn systems, for which the integer feasibility is known to be NP-complete~\cite{Lag85}.   
		
	In this paper, we study certifying  algorithms for LP problems with gainfree Leontief substitution systems.
	Recently, much attention has been paid to certifying algorithms in software engineering; see~\cite{MMN11} for a survey.
	Intuitively, an algorithm is called \emph{certifying} if it produces not only an answer but also a certificate with which we can easily confirm that the answer is correct.
	For the shortest $s$-$t$ path problem with positive edge length, 
	the potential of vertices (i.e., distances from $s$) is a certificate of a shortest $s$-$t$ path.
	Certifying algorithms have great advantages in practice because many commercial programs are reported to contain bugs~\cite{MMN11}.
	Certifying algorithms have been proposed for various problems in mathematical optimization and computer science~\cite{MNN99,DFK03,KMM06,KaN09,Sch13,CDH13,GeT15,CGS16,SuW17,MNS17,BGM22}.
	
	Let us briefly summarize certifying algorithms related to gainfree Leontief substitution systems.
	Standard LP solvers output a certificate of the optimality of an optimal solution; however, no combinatorial and strongly polynomial time algorithm for general LP problems is known and algorithms that work for special types of LP problems have been extensively studied.
	We first note that the well-known Bellman-Ford algorithm for the shortest path problem allowing negative edge length can be regarded as a certifying algorithm for the feasibility of DC systems.
	In fact, the algorithm computes a feasible solution which correspond to the potential of the associated graph $G$ if it is feasible, and a minimal infeasible subsystem that corresponds to a negative cycle in $G$ if it is infeasible.
	This result was extended to the 
	unit-two-variable-per-inequality (UTVPI) systems, where a system is called \emph{unit-two-variable-per-inequality} if each inequality is of the form $\pm x_i \pm x_j \le c$ for some integer $c$.
	Min\'{e}~\cite{Min06} proposed a certifying algorithm for the feasibility of UTVPI systems by transforming such systems to DC systems.
	Therefore, the feasibility of the systems admits combinatorial ${\rm O}(mn)$ certifying algorithms.
	We note that the feasibility coincides with the integer feasibility for DC systems while it is not the case for UTVPI systems.
	A combinatorial ${\rm O}(mn+n^2\log n)$-time certifying algorithm for the integer feasibility of UTVPI systems were proposed by 
	Lahiri and Musuvathi~\cite{LaM05}.
	Gupta~\cite{Gup14} reported that a certifying algorithm exists for the feasibility of unit Horn systems with nonpositivity constraints on variables,\footnote{The current form of the algorithm and the proofs of its validity in~\cite{Gup14} contains several flaws; see a detailed discussion in Section~\ref{sec:main-algorithms}.} and 
	mentioned that it is open whether the feasibility problem admits certifying algorithms when the systems are unit Horn (without nonpositivity constraints) and unit-positive Horn~\cite{Gup14}.
	For LP problems with gainfree Leontief substitution systems, Jeroslow et al. proposed a combinatorial ${\rm O}(m^2n)$-time certifying algorithm when it has an optimal solution~\cite{JMR92}. 
	However, it remains open whether such LP problems admit a combinatorial and strongly polynomial time certifying algorithm when it has no optimal solution.
	
	\paragraph{Our contribution}
	In this paper, we propose a combinatorial ${\rm O}(m^3n)$-time certifying algorithm for LP problems with gainfree Leontief substitution systems
	when the LP problem has no optimal solution, i.e., when 
	it is unbounded or infeasible.
	This together with the algorithm by Jeroslow et al. provides a combinatorial ${\rm O}(m^3n)$-time fully certifying algorithm for LP problems with gainfree substitution systems.
	As a corollary of our result, we resolve the open problem for the feasibility and the integer feasibility for unit-positive Horn systems.
	
	Certifying infeasibility draws much attention in, e.g., the field of logic and it is open how to make existing successive-approximation type combinatorial algorithms (e.g., \cite{Glo64,JMR92,ChS13}) certifying for a fundamental class of unit Horn systems.
	In successive-approximation type algorithms, the values of variables are iteratively updated according to the constraints.
	Indeed, for DC systems, it is sufficient to store the previous edge (or constraint) that causes the value update of a variable.
	However, in unit Horn systems, this is not enough: We have to store all the history of the value updates of the variables\footnote{This seems the essential error in~\cite{Gup14}.}.
	Our algorithm stores in which iteration the values of variables are updated and how the values can be derived by the given constraints.
	This enables it to compute a certificate of dual infeasibility.
	Our algorithm also introduces a symbol representing an ``arbitrary large'' number so that it can compute a certificate of primal infeasibility.
	
	Our algorithm is based on the hypergraph representation of Leontief substitution systems introduced by Jeroslow et al.~\cite{JMR92}, 
	and computes a certificate based on Farkas' lemma, called 
	a Farkas' certificate, which was also used by Gupta~\cite{Gup14} for unit Horn systems with nonpositivity constraints\footnote{Gupta dealt with a linear system of the form $A^Ty \ge c, y \ge 0$.
		Since $y$ is a feasible solution of $A^Ty \ge c, y \ge 0$ if and only if it is a feasible solution of $A^T(-y) \le -c, -y \le 0$, 
		Gupta's algorithm can be modified so that it deals with the feasibility of $A^Ty \le c, y \le 0$ with unit Horn $A^T$.}.
	Moreover, our algorithm for the dual feasibility can be seen as an extension of the Bellman-Ford algorithm for the feasibility of DC systems.
	In fact, if a DC system is given, then our algorithm finds a feasible solution if it is feasible, and a minimal infeasible subsystem that corresponds to a negative cycle in the associated graph if it is infeasible, which is the same as the Bellman-Ford algorithm.
		
	As a generalization of the integer feasibility of unit-positive Horn systems, we consider the one of the dual of gainfree Leontief substitution systems.
	We first point out that it is NP-complete. 
	We then propose an integer version of our algorithm for the feasibility and show that it requires exponential time in the worst case.
	We also consider the integer feasibility of the (primal) gainfree Leontief substitution systems and show that it is NP-complete. These results provide a threshold between general gainfree Leontief substitution systems and unit-positive Horn systems for integer feasibility of both primal and dual LP problems.
	
	Recall that it is known the feasibility of primal and dual LP problems with gainfree Leontief substitution systems can be solved combinatorially in polynomial time, which is certifying for feasibility.
	However, certifying infeasibility was open.
	One might think that we can obtain a combinatorial certifying algorithm for these feasibility problems by incorporating the idea of two-phase simplex method and using a non-certifying combinatorial algorithm for these problems.
	The idea of two-phase simplex method is to transform the feasibility into an LP problem which always has an optimal solution, and hence the existing algorithm only certifying for feasible problems seems to apply.
	Actually, we confirm that this idea works for the feasibility of primal LP problems since the transformed LP problem is again an LP problem with gainfree Leontief substitution systems.
	However, 
	the constraint system of the transformed LP problem for the feasibility of a dual LP problem is no more gainfree Leontief substitution system, and thus the existing combinatorial algorithm does not apply.
	See Subsection \ref{subsec:two-phase-does-not-work} for details.
	
	\paragraph{Related work}
	Certifying infeasibility (unsatisfiability) of given constraints has attracted much attention in the literature, especially in proof theory.
	For example, in SAT, resolution refutation provides a certificate of the infeasibility of an unsatisfiable CNF formula, and it has been extensively studied in logic and theoretical computer science.
	However, the length of resolution refutation is exponential in the input size in the worst case.
	For integer linear systems, cutting plane methods provide a certificate of the infeasibility for integer linear systems, which has also exponential size in the worst case.
	See \cite{CGS17} for an exponential time certifying algorithm for general mixed integer linear programming problems.
	Cutting plane refutation for unit Horn systems has also been investigated \cite{KWS19}.
	Finally, we remark that the primal-dual methods for optimization problems are certifying (also for feasibility).
	
	\paragraph{Outline}
	The rest of the paper is organized as follows.
	Section~\ref{sec:preliminaries} formally defines our problem and introduces the same hypergraph representation of Leontief substitution systems as in~\cite{JMR92}.
	Section~\ref{sec:main-algorithms} provides our main algorithm, i.e., a combinatorial certifying algorithm for LP problems with gainfree Leontief substitution systems.
	Section~\ref{sec:discussions} discusses the integer versions of the primal and the dual LP problems and the two-phase method.
	Section~\ref{sec:conclusion} concludes the paper.
	
	\section{Preliminaries}
	\label{sec:preliminaries}
	
	Let 
	$\mathbb{R}$, $\mathbb{R}_{+}$, and $\mathbb{R}_{++}$ denote the sets of 
	reals, nonnegative reals, and positive reals, respectively.
	For positive integers $m$ and $n$, 
	a matrix $A \in \mathbb{R}^{m \times n}$ is called \emph{Leontief} if each column contains at most one positive entry.
	In this paper, it is always assumed that the positive elements of $A$ are all ones unless otherwise stated, since it is sufficient for our purpose as stated below.
	Let $A \in \mathbb{R}^{m \times n}$ be an $m \times n$ matrix,
	and let $\bs{b} \in \mathbb{R}^m$ be a vector of dimension $m$.
	A set of linear inequalities 
	\begin{align*}
		A\bs{x} = \bs{b}\ \  \text{and}\ \ \bs{x} \in \mathbb{R}^n_{+}
	\end{align*}
	is called a \emph{Leontief substitution system} if 
	$A$ is Leontief and $\bs{b} \ge \bs{0}$.
	
	In this paper, we consider the following linear programming (LP) problem: 
	\begin{align}
		\label{eq:LP}
		\begin{array}{ll}
			\rm{minimize} & \bs{c}^T\bs{\primal}\\
			\rm{subject\ to} & A\bs{\primal} = \bs{b}\\
			& \bs{\primal} \in \mathbb{R}^n_+,
		\end{array}
	\end{align}
	where the constraint is a Leontief substitution system and $\bs{c} \in \mathbb{R}^n$.
	As stated above, 
	we assume throughout the paper that the positive elements of $A$ are all ones unless otherwise stated, 
	since otherwise it can be obtained by scaling the variables in the LP problem~\eqref{eq:LP} with a Leontief substitution system.
	
	We particularly focus on the subclass of LP with Leontief substitution systems satisfying the \emph{gainfree} property.
	To define gainfreeness, it is convenient to introduce a hypergraph representation~\cite{JMR92} of Leontief substitution systems.
	This representation is also used to state our algorithms.
	
	A hypergraph $\mathcal{H}$ is an ordered pair $\mathcal{H} = (V,\mathcal{E})$, where $V$ is a finite set called a vertex set and $\mathcal{E}$ is a set of hyperarcs.
	A hyperarc $E \in \mathcal{E}$ is an ordered pair $(H(E),T(E))$ of its head and tail sets, 
	where $H(E),T(E) \subseteq V$ and $H(E) \cap T(E) =\emptyset$.
	In our use, $|H(E)|$ is always at most one, i.e., $|H(E)| \le 1$.
	Hence, we denote $H(E)$ by $h(E)$, and when $|h(E)| = 1$, we identify $h(E)$ with the unique element in $h(E)$, e.g., if $v \in h(E)$, then we write $h(E)=v$.
	
	Now, we explain how to define an associated hypergraph $\mathcal{H} = (V,\mathcal{E})$ from a given LP problem with a Leontief substitution system \eqref{eq:LP}.
	For a positive integer $k$, let $[k] = \{1, \dots, k\}$.
	Let $V = \{ v_i \mid i \in [m] \}$, where $v_i$ corresponds to the $i$th row of $A$ in \eqref{eq:LP} for $i \in [m]$, 
	and let $\mathcal{E} = \{ E_j \mid j \in [n] \}$, where for each $j \in [n]$ a hyperarc $E_j$ is defined as $h(E_j) = v_i$ if $A_{ij} = 1$ for some $i \in [m]$ and $h(E_j) = \emptyset$ otherwise (i.e., $A_{ij} \le 0$ for all $i \in [m]$), and $T(E_j) = \{ v_i \in V \mid A_{ij} < 0 \}$.
	Note that for each $j \in [n]$ hyperarc $E_j$ corresponds to variable $x_j$ in \eqref{eq:LP}.
	We also associate a length function $\ell:\mathcal{E} \rightarrow \mathbb{R}$ to the hyperarc set $\mathcal{E}$, where $\ell(E_j) = c_j$ for each $E_j \in \mathcal{E}$.
	Moreover, we associate a positive value to each element of the tails of the hyperarcs in $\mathcal{E}$, namely, $\weight:\bigcup_{j \in [n]}(\{ E_j \} \times T(E_j)) \rightarrow \mathbb{R}_{++}$ defined as $\weight(E_j,v_i) = -A_{ij}\,(>0)$ for each $E_j \in \mathcal{E}$ and $v_i \in T(E_j)$.
	Note that the hypergraph is defined by matrix $A$ and vector $\bs{c}$ (and $\bs{b}$ is irrelevant).
	
	\begin{example}
		\label{ex:hypergraph}
		For the following input data, the associated hypergraph is drawn in \cref{fig:hypergraph2-1}.
		
		\begin{align}
			\label{eq:ex-hypergraph-input}
			A=
			\begin{pmatrix}
				-(1/2) & 0 & 1 & 1 & 0 \\
				1 & -(1/3) & 0 & 0 & 0 \\
				0 & 1 & -9 & 0 & 1 \\
				-(1/3) & -3 & -1 & 0 & 0 \\
			\end{pmatrix}\ and \ 
				c=\begin{pmatrix}
					-6\\5\\3\\-4\\2
				\end{pmatrix}.
			\end{align}
			\begin{figure}[t]
				\centering
				\includegraphics[scale=0.6, bb=0 0 200 150]{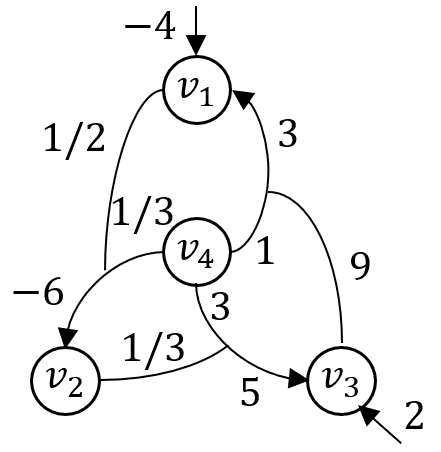}
				\caption{The hypergraph representation corresponding to the input~\eqref{eq:ex-hypergraph-input}}
				\label{fig:hypergraph2-1}
			\end{figure}
			\if0 
			\begin{align}
				\label{eq:ex-hypergraph}
				A^Ty \le c \equiv \left\{
				\begin{array}{rl}
					-(1/2)y_1 + y_2 - (1/3)y_4 &\le -6\\
					y_3 - (1/3)y_2 - 3y_4 &\le 5\\
					y_1 - 9y_3 - y_4 &\le 3\\
					y_1&\le -4\\
					y_3&\le 2.\\
				\end{array}
				\right.
			\end{align}
			\fi 
		\end{example}
		
		A \emph{directed path} in hypergraph $\mathcal{H}$ from vertex $v_1$ to $v_{k+1}$ is defined by a nonempty sequence $$v_1E_1v_2E_2v_3,\dots, E_k v_{k+1},$$ whose terms are alternatively vertices and hyperarcs, 
		with no intermediate vertex or hyperarc repeated, 
		such that $v_{i+1} = h(E_i)$ and $v_{i} \in T(E_i)$ for $i = 1, \dots, k$.
		A directed path from vertex $v_1$ to $v_{k+1}$ is a \emph{directed cycle} if $v_1 = v_{k+1}$.
		
		Now, we are ready to define gainfreeness.
		
		\begin{definition}[Gainfreeness]
			Let $v_1E_1v_2E_2v_3,\dots, E_k v_{k+1}$ be a directed cycle, where $v_1= v_{k+1}$. 
			The gain of this directed cycle is defined by
			\begin{align*}
				\frac{1}{\prod_{i=1}^k \gamma(E_i,v_i)}.
			\end{align*}
			We term a Leontief substitution system (and its defining matrix) \emph{gainfree} if the gain of every directed cycle in the associated hypergraph is at most one.
		\end{definition}
		
		From definition, 
		unit and unit-positive Horn matrices are transpose of gainfree Leontief matrices.
		
		\begin{example}
			In \cref{ex:hypergraph}, 
			the unique directed cycle of the hypergraph representation is
			$v_1E_1v_2E_2v_3E_3v_1$, where each $E_i$ corresponds to the $i$th inequality.
			The gain of this cycle is $1/(1/2 \cdot 1/3 \cdot 9) = 2/3 \le 1$.
			Hence, matrix $A$ in \eqref{eq:ex-hypergraph-input} is gainfree.
		\end{example}
		
		Now, we recall some notion from LP theory.
		A vector $\bs{x} \in \mathbb{R}^n_{+}$ is called a \emph{feasible solution} of \eqref{eq:LP} if it satisfies the inequalities in \eqref{eq:LP}.
		An LP problem is \emph{feasible} if it has a feasible solution, and \emph{infeasible} otherwise.
		A vector $\bs{x} \in \mathbb{R}^n_{+}$ is called an \emph{optimal solution} of \eqref{eq:LP} if it is feasible and $\bs{c}^T\bs{x} \le \bs{c}^T\bs{x}'$ for any feasible solution $\bs{x}'$.
		When an LP problem has an optimal solution $\bs{x}$, 
		the objective value $\bs{c}^T\bs{x}$ is called an \emph{optimal value}.
		An LP problem is either feasible or infeasible, and when it is feasible either it has an optimal solution or it is unbounded (i.e., its optimal value is not bounded below).
		Since we consider certifying algorithms, we have to produce a certificate in each case.
		To state what constitutes a certificate in each case, 
		we recall the \emph{dual} LP problem of \eqref{eq:LP}: 
		\begin{align}
			\label{eq:dual-LP}
			\begin{array}{ll}
				\rm{maximize} & \bs{\dual}^T\bs{b}\\
				\rm{subject\ to} & \bs{\dual}^TA \leq \bs{c}^T\\
				& \bs{\dual} \in \mathbb{R}^m.
			\end{array}
		\end{align}
		To contrast, the LP problem~\eqref{eq:LP} is called the \emph{primal} LP problem in what follows.

		The following duality theorem of LP is well-known.
		\begin{theorem}[E.g., \cite{Sch98}]
			\label{thm:LP-solution-pattern}
			For the LP problem \eqref{eq:LP} and its dual problem \eqref{eq:dual-LP}, exactly one of the following holds: 
			\begin{description}
				\item[(i)] 
				both \eqref{eq:LP} and \eqref{eq:dual-LP} have feasible solutions whose objective values are the same, 
				\item[(ii)] 
				\eqref{eq:LP} is infeasible, and \eqref{eq:dual-LP} feasible and unbounded, 
				\item[(iii)]  
				\eqref{eq:LP} is feasible and unbounded, and \eqref{eq:dual-LP} is infeasible;
				\item[(iv)] 
				both \eqref{eq:LP} and \eqref{eq:dual-LP} are infeasible.
			\end{description}
		\end{theorem}
		
		We regard a feasible solution as a certificate of feasibility of an LP problem.
		For infeasibility we use the following lemmas.
		
		\begin{lemma}[Farkas' lemma for the primal infeasibility (e.g., \cite{Sch98})]
			\label{lem:Farkas-primal}
			For positive integers $m$ and $n$, 
			let $A \in \mathbb{R}^{m \times n}$ be a matrix and $\bs{b} \in \mathbb{R}^m$ be a vector.
			The linear system 
			\begin{align*}
				\left\{
				\begin{array}{l}
					A\bs{x} = \bs{b},\\
					\bs{x} \in \mathbb{R}_+^n
				\end{array}
				\right.
			\end{align*}
			is infeasible if and only if 
			\begin{align}\label{eq:Farkas-certificate}
				\left\{
				\begin{array}{l}
					\bs{\dray}^TA \le \bs{0},\\
					\bs{\dray}^T\bs{b} > 0,\\
					\bs{\dray} \in \mathbb{R}^m
				\end{array}
				\right.
			\end{align}
			is feasible.
		\end{lemma}
		
		The infeasibility of the dual LP problem \eqref{eq:dual-LP} is characterized as follows.
		
		\begin{lemma}[Farkas' lemma for the dual infeasibility (e.g., \cite{Sch98})]
			\label{lem:Farkas-dual}
			Let $m,n$ be positive integers.
			Let $A \in \mathbb{R}^{m \times n}$ be a matrix and $\bs{c} \in \mathbb{R}^n$ be a vector.
			The linear system 
			\begin{align*}
				\left\{
				\begin{array}{l}
					\bs{\dual}^TA \le \bs{c}^T,\\
					\bs{\dual} \in \mathbb{R}^m
				\end{array}
				\right.
			\end{align*}
			is infeasible if and only if 
			\begin{align}
				\label{eq:dual-Farkas-certificate}
				\left\{
				\begin{array}{l}
					A\bs{\pray} = \bs{0},\\
					\bs{c}^T\bs{\pray} < 0,\\
					\bs{\pray} \in \mathbb{R}_+^n
				\end{array}
				\right.
			\end{align}
			is feasible.
		\end{lemma}
		
		Now, we define what constitute certificates for the four possible cases in \cref{thm:LP-solution-pattern}.
		
		\begin{description}
			\item[(i)] Feasible solutions of \eqref{eq:LP} and \eqref{eq:dual-LP} whose objective values are the same, 
			\item[(ii)] a feasible solution of \eqref{eq:Farkas-certificate} (called a \emph{Farkas' certificate} of infeasibility of \eqref{eq:LP}) and 
			a feasible solution of \eqref{eq:dual-LP},
			\item[(iii)] a feasible solution of \eqref{eq:LP} and a feasible solution of \eqref{eq:dual-Farkas-certificate} (called a \emph{Farkas' certificate} of infeasibility of \eqref{eq:dual-LP}),
			\item[(iv)] a feasible solution of \eqref{eq:Farkas-certificate} and a feasible solution of \eqref{eq:dual-Farkas-certificate}.
		\end{description}

		With those certificates, we can confirm the correctness of the output of our certifying algorithm for solving the LP problem~\eqref{eq:LP} by checking if given vectors satisfy the corresponding linear systems.
		We note that for case (ii) (resp., (iii)) a feasible solution of \eqref{eq:Farkas-certificate} (resp., \eqref{eq:dual-Farkas-certificate}) is a direction of unboundedness.
		
		Finally, we summarize the notations used throughout the paper.
		For $I \subseteq [m]$ and $J \subseteq [n]$, let $A_{I,J}$ be the submatrix of $A$ whose rows and columns are restricted to $I$ and $J$, respectively.
		We sometimes denote $A_{[m],J}$ (resp., $A_{I,[n]}$) by $A_{.J}$ (resp., $A_{I.}$).
		If $I \subseteq [m]$ (resp., $J \subseteq [n]$) is a singleton set, e.g., $I=\{i\}$ (resp., $J=\{j\}$), we denote $A_{\{i\},J}$ (resp., $A_{I,\{j\}}$) by $A_{i,J}$ (resp., $A_{I,j}$).
		These rules apply simultaneously, e.g., if $I = [m]$ and $J = \{j\}$, we denote $A_{.j}$ (which is the $j$th column vector of $A$). 
		For $I \subseteq [m]$ (resp., $J \subseteq [n]$), let $\overline{I} = [m] \setminus I$ (resp., $\overline{J} = [n] \setminus J$).
		For an $m$-dimensional vector $\bs{b}$ and $I \subseteq [m]$, let $\bs{b}_{I}$ be the vector obtained by restricting the coordinates of $\bs{b}$ to $I$.
		We denote by $e_i$ an unit vector of appropriate size, where its
		$i$th element is $1$ and all other elements are $0$.
		
		\section{Main algorithms}
		\label{sec:main-algorithms}
		
		In this section, 
		we provide a combinatorial certifying algorithm for LP problems with gainfree Leontief substitution systems \eqref{eq:LP} 
		and show the following theorem.
		Here, a combinatorial algorithm consists only of additions, subtractions, multiplications, and comparisons.
		Recall that $m$ is the number of constraints and $n$ is the number of variables in \eqref{eq:LP}.

		\begin{theorem}[Main]
			\label{thm:main-theorem}
			The LP problems with gainfree Leontief substitution systems \eqref{eq:LP} admit a combinatorial 
			${\rm O}(m^3n)$-time 
			certifying algorithm. 
		\end{theorem}
		
		Our combinatorial certifying algorithm for the LP problems with gainfree Leontief substitution systems \eqref{eq:LP} is an extension of the non-certifying algorithm in~\cite{JMR92}.
		Let us first summarize the non-certifying algorithm in~\cite{JMR92}, which consists of \textsc{ValueIteration} and \textsc{PrimalRetrieval}. \textsc{ValueIteration} determines the feasibility of the dual LP problem~\eqref{eq:dual-LP}. 
		It starts from a sufficiently large vector and iteratively compute an upper bound of the value of each variable derived from the constraints in \eqref{eq:dual-LP}.
		For an LP problem with a gainfree Leontief substitution system, $m$ iterations is shown to be sufficient to obtain a feasible solution if the dual LP problem is feasible.
		Then, the feasibility of the primal LP problem~\eqref{eq:LP} can be determined using the data computed in \textsc{ValueIteration}, and when it is feasible, \textsc{PrimalRetrieval} computes a feasible solution of it.
		This algorithm outputs feasible solutions of the primal and dual LP problems with the same objective values as a certificate of primal and dual feasibility for case (i) in \cref{thm:LP-solution-pattern} in \cref{sec:preliminaries}.
		
		To make the algorithm in~\cite{JMR92} also certifying for the primal and dual infeasibility (i.e., for cases (ii-iv) in \cref{thm:LP-solution-pattern}), we modify the algorithm and add several subroutines to it.
		We first modify \textsc{ValueIteration} to \textsc{DualFeasibility} (\cref{alg:dual-feasibility}).
		In \textsc{DualFeasibility}, when the upper bound $\bs{\dual}^{(k)}$ for the dual variables is updated in the $k$th iteration of the for-loop starting from line 2, we store variables changed in the iteration and a vector $\bs{\pray}^{(k)}$, which represents how an upper bound $\bs{\dual}^{(k)}$ is derived from the constraint in \eqref{eq:dual-LP}.
		This enables us to compute a Farkas' certificate of dual infeasibility in \textsc{FarkasCertificateOfDualInfeasibility} (\cref{alg:primal-ray}) when the dual LP problem is infeasible.
		This modification also makes our algorithm different from the one in~\cite{Gup14}.
		Since the upper bound $\bs{\dual}^{(m)}$ computed in \textsc{DualFeasibility} contains symbol $M$ as described below, we compute 
		in \textsc{DualSolution} (\cref{alg:dual-solution}) a feasible solution of the dual LP problem from $\bs{\dual}^{(m)}$ when the dual LP problem is feasible.
		\textsc{PrimalFeasibility} (\cref{alg:primal-feasibility}) determines the feasibility of the primal LP problem~\eqref{eq:LP} using the same criterion as in (ii) of Theorem 3.6 in \cite{JMR92}.
		\textsc{PrimalSolution} (\cref{alg:primal-solution}) is different from \textsc{PrimalRetrieval} in~\cite{JMR92} 
		in that the former computes a primal feasible solution not only when the dual LP problem is feasible but also it is infeasible.
		Finally, in \textsc{DualFeasibility} we treat $M$ as a symbol representing an ``arbitrary large'' number so that we can compute a Farkas certificate of primal infeasibility in \textsc{FarkasCertificateOfPrimalInfeasibility} (\cref{alg:dual-ray}).
		For any real numbers $\alpha_1,\alpha_2,\beta_1,\beta_2 \in \mathbb{R}$, 
		we define $\alpha_1 M + \beta_1 > \alpha_2 M + \beta_2$ if and only if $\alpha_1 > \alpha_2$ or ($\alpha_1 = \alpha_2$ and $\beta_1 > \beta_2$).
		
		For the readability, we first describe a certifying algorithm for the feasibility of the dual of the LP problems (with gainfree Leontief substitution systems) in~\cref{subsec:dual-feasibility} and one for the feasibility of the primal LP problems in~\cref{subsec:primal-feasibility}.
		A proof of \cref{thm:main-theorem} will be given in~\cref{subsec:proof-of-main-theorem}.
			
		\subsection{A certifying algorithm for the feasibility of the dual LP problem}
		\label{subsec:dual-feasibility}
		
		In this subsection, we provide a certifying algorithm for the feasibility of the dual~\eqref{eq:dual-LP} of the LP problem with a gainfree Leontief substitution system.
		The main algorithm (\cref{alg:cert-gainfree-Leontief-dual}) first calls subroutine \textsc{DualFeasibility} (\cref{alg:dual-feasibility}), which determines the feasibility of the dual LP problem~\eqref{eq:dual-LP}.
		If it is feasible, then subroutine \textsc{DualSolution} (\cref{alg:dual-solution}) is called to compute a feasible solution of the dual LP problem; otherwise, 
		subroutine \textsc{FarkasCertificateOfDualInfeasibility} (\cref{alg:primal-ray}) is called to compute a Farkas' certificate of the dual infeasibility.
		
		\begin{algorithm}
			\caption{Combinatorial certifying algorithm for the feasibility of the dual of the LP problems with gainfree Leontief substitution systems}
			\label{alg:cert-gainfree-Leontief-dual}
			\KwInput{A matrix $A$ and a vector $\bs{c}$ for the constraint of the dual LP problem~\eqref{eq:dual-LP}.}
		($\bs{y}^{(m)},\bs{r}^{(m)},\change^{(k)}(k=0,...,m),p^{(k)}(k=0,...,m),\nontriv^{(m)},\bs{q},\val$)$\ot$\textsc{DualFeasibility}($A,\bs{c}$). \\
		\eIf{$\val=\true$}
		{ \textsc{DualSolution}.\\
			$\bs{\dual}^* \ot$ {\rm \textsc{DualSolution}}($A,\bs{c},\bs{y}^{(m)}$).\\
			{\bf print} ``dual-feasible'' and {\bf return} $\bs{\dual}^*$.} 
	{		$\bs{\pray}^* \ot$ {\rm \textsc{FarkasCertificateOfDualInfeasibility}}($A,\bs{c},\bs{y}^{(m)},\bs{r}^{(m)},\change^{(k)}(k=0,...,m),p^{(k)}(k=0,...,m)$).\\
		{\bf print} ``dual-infeasible'' and {\bf return} $\bs{\pray}^*$.
	}
\end{algorithm}	

\begin{algorithm}
	\caption{\textsc{DualFeasibility}}
	\label{alg:dual-feasibility}
	\KwInput{A matrix $A$ and a vector $\bs{c}$ for the constraint of the dual LP problem~\eqref{eq:dual-LP}.} 
	For each $v \in V$, $\dual^{(0)}(v)\ot M$, $\bs{\pray}_v^{(0)}\ot \bs{0}$, $\change^{(0)}(v)\ot \false$, $p^{(0)}(v)\ot \emptyset$, $\nontriv^{(0)}(v)\ot \false$, and $q(v)\ot 0$.\\
	\For{$k=1,\dots,m$}{
		\For{$v \in V$}{
			\eIf{$\dual^{{(k-1)}}(v) > \min\left\{ \ell(E) + \sum_{u \in T(E)}\weight(E,u)\dual^{(k-1)}(u) \mid E \in \mathcal{E}, h(E) = v \right\}$}
			{Choose an arbitrary $E \in \argmin\left\{ \ell(E) + \sum_{u \in T(E)}\weight(E,u)\dual^{(k-1)}(u) \mid E \in \mathcal{E}, h(E) = v \right\}$.\\
				$\dual^{(k)}(v)\ot \ell(E) + \sum_{u \in T(E)}\weight(E,u)\dual^{(k-1)}(u)$.\\
				$p^{(k)}(v)\ot E$.\\
				$\bs{\pray}^{(k)}_v\ot e_{E} + \sum_{u\in T(E)}\weight(E,u)\bs{\pray}^{(k-1)}_{u}$.\\
				$\change^{(k)}(v)\ot \true$.\\
				\eIf{for every $u \in T(E)$ $\nontriv^{(k-1)}(u) = \true$ (this includes the case that $T(E)=\emptyset$)}
				{$\nontriv^{(k)}(v) \ot \true$ and $q(v)\ot k$.}
				{$\nontriv^{(k)}(v) \ot \nontriv^{(k-1)}(v)$.}
			}
			{$\dual^{(k)}(v)\ot \dual^{(k-1)}(v)$, $p^{(k)}(v)\ot \emptyset$, $\bs{\pray}^{(k)}_v\ot \bs{\pray}^{(k-1)}_v$, $\change^{(k)}(v)\ot \false$, and 
				$\nontriv^{(k)}(v)\ot \nontriv^{(k-1)}(v)$.}
		}
	}
	\uIf{$\dual^{(m)}(v) > \min\left\{ \ell(E) + \sum_{u \in T(E)}\weight(E,u)\dual^{(m)}(u) \mid E \in \mathcal{E}, h(E) = v \right\}$ {\rm for some} $v \in V$}
	{		$\val\ot \false$.}
	\uElseIf{$0 > \ell(E) + \sum_{u \in T(E)}\weight(E,u)\dual^{(m)}(u)$ for some $E \in \mathcal{E}$ with $h(E)=\emptyset$}{
		$\val\ot \false$.}
	\Else{
		$\val\ot \true$.}
	{\bf return} ($\bs{y}^{(m)},\bs{r}^{(m)},\change^{(k)}(k=0,...,m),p^{(k)}(k=0,...,m),\nontriv^{(m)},\bs{q},\val$).
\end{algorithm}	

\begin{algorithm}
	\caption{
		\textsc{\textsc{DualSolution}}
	}
	\label{alg:dual-solution}
	\KwInput{A matrix $A$ and a vector $\bs{c}$ for the constraint of the dual LP problem~\eqref{eq:dual-LP}, and an $n$-dimensional vector $\bs{y}$ with each entry being a linear function of $M$.}
	\For{each $E \in \mathcal{E}$}{Define two integers $\alpha(E)$ and $\beta(E)$ such that $\alpha(E)M + \beta(E) = \dual^{(m)}(h(E)) - \ell(E) - \sum_{u \in T(E)}\weight(E,u)\dual^{(m)}(u)$ , if where we define $\dual^{(m)}(\emptyset) = 0$.}
	\eIf{all $E \in \mathcal{E}$ satisfy $\alpha(E) \ge 0$}
	{$\lambda \ot 0$.}
	{$\lambda \ot \max\left\{ \frac{\beta(E)}{-\alpha(E)} \mid E \in \mathcal{E}, \alpha(E) < 0 \right\}$.}
	Let $\bs{\dual}^*$ be the vector obtained from $\bs{\dual}$ by substituting $\lambda$ with $M$.\\
	{\bf return} $\bs{\dual}^*$.
\end{algorithm}	

\begin{algorithm}
	\caption{\textsc{FarkasCertificateOfDualInfeasibility}}
	\label{alg:primal-ray}
	\KwInput{A matrix $A$ and a vector $\bs{c}$ for the constraint of the dual LP problem~\eqref{eq:LP}, $\bs{y}^{(m)}$, $\bs{r}^{(m)}$, and $\change^{(k)}$ and $p^{(k)}$ for $k=0,...,m$.}
	\eIf{$\dual^{(m)}(v) > \min\left\{ \ell(E) + \sum_{u \in T(E)}\weight(E,u)\dual^{(m)}(u) \mid E \in \mathcal{E}, h(E) = v \right\}$ {\rm for some} $v \in V$}
	{Choose one $v\in V$ such that $\dual^{(m)}(v) > \min\left\{ \ell(E) + \sum_{u \in T(E)}\weight(E,u)\dual^{(m)}(u) \mid E \in \mathcal{E}, h(E) = v \right\}$.\\
		Choose an arbitrary $E \in \mathcal{E}$ with $h(E)=v$ that minimizes $\ell(E) + \sum_{u \in T(E)}\weight(E,u)\dual^{(m)}(u)$.\\
		$w_{m+1} \ot v$.\\
		$\bs{\pray}^{(m+1)}_{w_{m+1}} \ot  e_{E} + \sum_{u\in T(E)}\weight(E,u)\bs{\pray}^{(m)}_{u}$.\\
		$E^{(m+1)}\ot E$.\\
		/* Find a cycle */\\
		\For{$k=m+1,\dots, 2$}{
			Choose an arbitrary $u \in T(E^{(k)})$ such that $\change^{(k-1)}(u) = \true$.\\
			$w_{k-1} \ot u$.\\
			$E^{(k-1)} \ot p^{(k-1)}(w_{k-1})$.\\
			\If{$w_{k-1} = w_{q}$  for some $q \ge k$}
			{$t \ot  q$.\\
				$s \ot k-1$.\\ 
				Break.}
		}
		$\bs{\pray}^* \ot \bs{\pray}^{(t)}_{w_t} - \bs{\pray}^{(s)}_{w_s}$.\\
		{\bf return} $\bs{\pray}^*$.}
	{Choose one $E \in \mathcal{E}$ with $h(E) = \emptyset$ such that $0 > \ell(E) + \sum_{u \in T(E)}\weight(E,u)\dual^{(m)}(u)$.\\
		$\bs{\pray}^*\ot \bs{e}_E + \sum_{u \in T(E)}\weight(E,u)\bs{\pray}^{(m)}_{u}$.\\
		{\bf return} $\bs{\pray}^*$.
	}
\end{algorithm}	

Before going into the proofs of correctness of these algorithms, we show several examples how these algorithms work.
We only show how $\bs{y}^{(k)}$ and $\bs{\pray}^{(k)}_{v}$ are updated in each iteration of the for-loop starting from line 2 in \textsc{DualFeasibility} in these examples for readability.
Also, we omit the input vector $\bs{b}$ in these examples, since $\bs{b}$ is irrelevant to the feasibility of the dual LP problem~\eqref{eq:dual-LP}.

\begin{example}
	\label{ex:main-algo1}
	For the following input data, the associated hypergraph is drawn in \cref{fig:hypergraph}.
	
	\begin{figure}[t]
		\centering
		\includegraphics[scale=0.6,bb= 0 0 200 150]{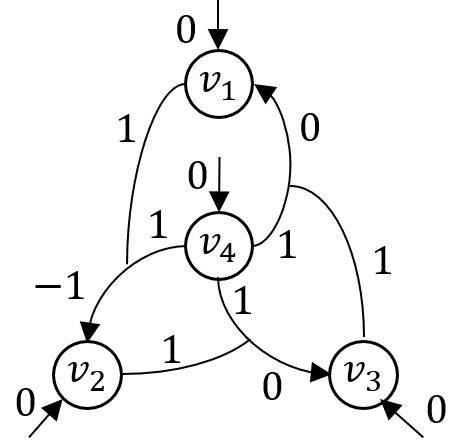}
		\caption{The hypergraph representation of system~\eqref{eq:ex-dual-infeas2}}
		\label{fig:hypergraph}
	\end{figure}
	
	\begin{align*}
		A=
		\begin{pmatrix}
			-1 & 0 & 1 & 1 & 0 & 0 & 0 \\
			1 & -1 & 0 & 0 & 1 & 0 & 0 \\
			0 & 1 & -1 & 0 & 0 & 1 & 0 \\
			-1 & -1 & -1 & 0 & 0 & 0 & 1 \\
		\end{pmatrix}\ and \ 
			c=\begin{pmatrix}
				-1\\0\\0\\0\\0\\0\\0
			\end{pmatrix}.
		\end{align*}
		
		The constraints in the dual LP problem~\eqref{eq:dual-LP} are 
		
		\begin{align}
			\label{eq:ex-dual-infeas2}
			A^Ty \le c \equiv \left\{
			\begin{array}{rl}
				-y_1 + y_2 - y_4 &\le -1\\
				-y_2 + y_3 - y_4 &\le 0\\
				y_1 - y_3 - y_4 &\le 0\\
				y_1&\le 0\\
				y_2&\le 0\\
				y_3&\le 0\\
				y_4&\le 0.
			\end{array}
			\right.
		\end{align}
		Initially, $\bs{\dual}^{(0)} = (M,M,M,M)$ and $\bs{\pray}^{(0)}_{v_i} = \bs{0}$ for $i=1,2,3,4$.\\
		Iteration 1: 
		$\bs{\dual}^{(1)} = (0,0,0,0)$ and 
		$\bs{\pray}^{(1)}_{v_i} = e_{i+3}$ ($i=1,2,3,4$).\\
		Iteration 2: 
		$\bs{\dual}^{(2)} = (0,-1,0,0)$ and 
		$\bs{\pray}^{(2)}_{v_2} = e_1+e_4+e_7$.\\
		Iteration 3: 
		$\bs{\dual}^{(3)} = (0,-1,-1,0)$ and 
		$\bs{\pray}^{(3)}_{v_3} = e_1+e_2+e_4+2e_7$.\\
		Iteration 4: 
		$\bs{\dual}^{(4)} = (-1,-1,-1,0)$ and 
		$\bs{\pray}^{(4)}_{v_1} = e_1+e_2+e_3+e_4+3e_7$.\\
		Now, 
		$(-1,0,1,-1)\bs{\dual}^{(4)} = 0 > -1$.
		Hence, the first inequality is violated by $\bs{\dual}^{(4)}$.
		Then, by running \textsc{FarkasCertificateOfDualInfeasibility}, we have 
		$\bs{\pray}^*=\bs{\pray}^{(5)}_{v_2}-\bs{\pray}^{(2)}_{v_2}=e_1+e_2+e_3+3e_7$.
		Thus, $A\bs{\pray}^*=(0,0,0,0)^T$ and $c^T\bs{\pray}^*=-1$.
		Hence, $\bs{\pray}^*$ is a Farkas certificate of infeasibility of system~\eqref{eq:ex-dual-infeas2}.
		
	\end{example}
	
	\begin{example}
		\label{ex:main-algo2}
		For the following input data
		\begin{align*}
			A=
			\begin{pmatrix}
				-2 & 1 & 0 & 0 \\
				-5 & -2 & 1 & 0 \\
				-3 & -1 & -2 & 1 \\
			\end{pmatrix}\ and \ 
				c=\begin{pmatrix}
					-64\\
					3\\
					1\\
					2
				\end{pmatrix},
			\end{align*}
			
			the constraints in the dual LP problem~\eqref{eq:dual-LP} are 
			
			\begin{align}
				\label{eq:ex-dual-infeas}
				A^Ty \le c \equiv \left\{
				\begin{array}{rl}
					-2y_1 - 5y_2 - 3y_3 &\le -64\\
					y_1 - 2y_2 - y_3 &\le 3\\
					y_2 - 2y_3 &\le 1\\
					y_3 &\le 2.
				\end{array}
				\right.
			\end{align}
			Initially, $\bs{\dual}^{(0)} = (M,M,M)$ and $\bs{\pray}^{(0)}_{v_i} = (0,0,0,0)$ for $i=1,2,3$.\\
			Iteration 1: 
			$\bs{\dual}^{(1)} = (M,M,2)$ and 
			$\bs{\pray}^{(1)}_{v_3} = (0,0,0,1)$.\\
			Iteration 2: 
			$\bs{\dual}^{(2)} = (M,5,2)$ and 
			$\bs{\pray}^{(2)}_{v_2} = (0,0,1,2)$.\\
			Iteration 3: 
			$\bs{\dual}^{(3)} = (15,5,2)$ and 
			$\bs{\pray}^{(3)}_{v_1} = (0,1,2,5)$.\\
			Now, $(-2,-5,-3)\bs{\dual}^{(3)} = -61 > -64$.
			Hence, the first inequality is violated by $\bs{\dual}^{(3)}$.
			Then, by running \textsc{FarkasCertificateOfDualInfeasibility}, we have $\bs{\pray}^*=(1,0,0,0)+2(0,1,2,5)+5(0,0,1,2)+3(0,0,0,1)=(1,2,9,23)$.
			Then $A\bs{\pray}^*=(0,0,0)^T$ and $c^T\bs{\pray}^*=-3$.
			Hence, $\bs{\pray}^*$ is a Farkas certificate of infeasibility of system~\eqref{eq:ex-dual-infeas}.
		\end{example}
		
		\begin{example}
			\label{ex:main-algo3}
			For the following input data
			\begin{align*}
				A=
				\begin{pmatrix}
					-(1/2) & 1 & 1 \\
					1 & -(1/3) & 0 \\
					0 & 0 & -6 \\
				\end{pmatrix}\ and \ 
					c=\begin{pmatrix}
						-3\\
						1\\
						-2
					\end{pmatrix},
				\end{align*}
				
				the constraints in the dual LP problem~\eqref{eq:dual-LP} are 
				
				\begin{align}
					\label{eq:ex-dual-infeas4}
					A^Ty \le c \equiv \left\{
					\begin{array}{rl}
						-(1/2)y_1 + y_2  &\le -3\\
						- (1/3)y_2 + y_3 &\le 1\\
						y_1 - 6y_3 &\le -2.
					\end{array}
					\right.
				\end{align}
				Initially, $\bs{\dual}^{(0)} = (M,M,M)$ and $\bs{\pray}^{(0)}_{v_i} = (0,0,0)$ for $i=1,2,3$.\\
				Iteration 1: 
				$\bs{\dual}^{(1)} = (M,(1/2)M-3,(1/3)M+1)$, 
				$\bs{\pray}^{(1)}_{v_2} = (1,0,0)$, and 
				$\bs{\pray}^{(1)}_{v_3} = (0,1,0)$.\\
				Iteration 2: 
				$\bs{\dual}^{(2)} = (M,(1/2)M-3,(1/6)M)$, 
				$\bs{\pray}^{(2)}_{v_3} = (1/3,1,0)$.\\
				Iteration 3: 
				$\bs{\dual}^{(3)} = (M-2,(1/2)M-3,(1/6)M)$, 
				$\bs{\pray}^{(3)}_{v_1} = (2,6,1)$.\\
				Now, $(-(1/2),1,0)\bs{\dual}^{(3)} = -2 > -3$.
				Hence, the first inequality is violated by $\bs{\dual}^{(3)}$.
				Then, by running \textsc{FarkasCertificateOfDualInfeasibility}, we have 
				$\bs{\pray}^*=\bs{\pray}^{(4)}_{v_2}-\bs{\pray}^{(1)}_{v_2}=(1,3,1/2)$.
				Thus, $A\bs{\pray}^*=(0,0,0)^T$ and $c^T\bs{\pray}^*=-1$.
				Hence, $\bs{\pray}^*$ is a Farkas certificate of infeasibility of system~\eqref{eq:ex-dual-infeas4}.
			\end{example}
			
			\begin{example}
				\label{ex:main-algo4}
				For the following input data
				\begin{align*}
					A=
					\begin{pmatrix}
						-(1/2) & 1 & 1 \\
						1 & -(1/3) & 0 \\
						0 & 0 & -6 \\
					\end{pmatrix}\ and \ 
						c=\begin{pmatrix}
							-3\\
							1\\
							2
						\end{pmatrix},
					\end{align*}
					
					the constraints in the dual LP problem~\eqref{eq:dual-LP} are 
					
					\begin{align}
						\label{eq:ex-dual-infeas5}
						A^Ty \le c \equiv \left\{
						\begin{array}{rl}
							-(1/2)y_1 + y_2  &\le -3\\
							- (1/3)y_2 + y_3 &\le 1\\
							y_1 - 6y_3 &\le 2.
						\end{array}
						\right.
					\end{align}
					Initially, $\bs{\dual}^{(0)} = (M,M,M)$ and $\bs{\pray}^{(0)}_{v_i} = (0,0,0)$ for $i=1,2,3$.\\
					Iteration 1: 
					$\bs{\dual}^{(1)} = (M,(1/2)M-3,(1/3)M+1)$, 
					$\bs{\pray}^{(1)}_{v_2} = (1,0,0)$, and 
					$\bs{\pray}^{(1)}_{v_3} = (0,1,0)$.\\
					Iteration 2: 
					$\bs{\dual}^{(2)} = (M,(1/2)M-3,(1/6)M)$, 
					$\bs{\pray}^{(2)}_{v_3} = (1/3,1,0)$.\\
					Iteration 3: 
					$\bs{\dual}^{(3)} = \bs{\dual}^{(2)} = (M,(1/2)M-3,(1/6)M)$. \\
					Now, all the inequalities are satisfied by $\bs{\dual}^{(3)}$.
					Then, by running \textsc{DualSolution}, we have 
					$\bs{\dual}^*=(0,-3,0)$.
					Then $\bs{\dual}^*$ is a certificate of feasibility (feasible solution) of system~\eqref{eq:ex-dual-infeas5}.
				\end{example}
				
				In the remainder of this subsection, 
				we will prove correctness of \cref{alg:cert-gainfree-Leontief-dual}.
				We show the correctness of subroutines \textsc{DualFeasibility}, \textsc{DualSolution}\footnote{\textsc{DualSolution} uses a division, however, we can avoid the division by using \textsc{ValueIteration} in \cite{JMR92} to obtain a feasible dual solution.}, and \textsc{FarkasCertificateOfDualInfeasibility}, and show the following proposition.
								
				\begin{proposition}
					\label{prop:dual-feasibility-correct}
					\cref{alg:cert-gainfree-Leontief-dual} is a combinatorial 
					${\rm O}(m^3n)$-time 
					certifying algorithm for the feasibility of the dual~\eqref{eq:dual-LP} of the LP problem with a gainfree Leontief substitution system.
				\end{proposition}
				
				As mentioned in Introduction, 
				the above proposition resolves the open questions raised in~\cite{Gup14}.
				
				To show \cref{prop:dual-feasibility-correct}, we first deal with the case where
				\cref{alg:cert-gainfree-Leontief-dual} prints ``dual-feasible'' (or, equivalently, \textsc{DualFeasibility} returns $\false$)
				in \cref{lem:dual-feasible} below.
				Then, we deal with the case where 
				\cref{alg:cert-gainfree-Leontief-dual} prints 
				``dual-infeasible'' (or, equivalently, \textsc{DualFeasibility} returns $\true$) in \cref{lem:dual-infeasible} below.
				
				\begin{lemma}\label{lem:dual-feasible}
					If 
					\textsc{DualFeasibility} returns $\true$, 
					then the dual LP problem~\eqref{eq:dual-LP} is feasible and 
					\textsc{DualSolution} outputs a feasible solution to it.
				\end{lemma}
				
				\begin{proof}
					We show that the output $\bs{\dual}^*$ of \textsc{DualSolution} is a feasible solution of the dual LP problem~\eqref{eq:dual-LP}.
					We divide the proof into cases according to the conditions in the definition of $\lambda$ in \textsc{DualSolution}.
					
					Fix $E \in \mathcal{E}$.
					Note that we have 
					\begin{align*}
						\dual^{(m)}(h(E)) \le \ell(E) + \sum_{u \in T(E)}\weight(E,u)\dual^{(m)}(u),
					\end{align*}
					since the conditions of ``if '' and ``else if '' in lines 20 and 22, respectively, are false in \textsc{DualFeasibility}, 
					where we define $\dual^{(m)}(\emptyset)=0$.
					Hence, we have $\alpha(E)M + \beta(E) \le 0$.
					It follows that $\alpha(E) \le 0$.
					If $\alpha(E) = 0$, then $\beta(E) \le 0$ and $\bs{\dual}^*$ satisfy the constraint in the dual LP problem~\eqref{eq:dual-LP} corresponding to $E$.
					If $\alpha(E) < 0$, then 
					$\bs{\dual}^*$ also satisfies the inequality in the dual LP problem~\eqref{eq:dual-LP} corresponding to $E$, since 
					$\lambda \ge \frac{\beta(E)}{-\alpha(E)}$ by definition.
					This completes the proof.
				\end{proof}
				
				Next, we treat the case where 
				\textsc{DualFeasibility} returns $\false$ 
				and show the following. 
				\begin{lemma}\label{lem:dual-infeasible}
					If 
					\textsc{DualFeasibility} returns $\false$, 
					then the dual LP problem~\eqref{eq:dual-LP} is infeasible and 
					\textsc{FarkasCertificateOfDualInfeasibility} returns a  Farkas' certificate of the dual infeasibility.
				\end{lemma}
				
				The proof of \cref{lem:dual-infeasible} is the most technical part of our results.
				Intuitively, when 
				\textsc{DualFeasibility} returns $\false$, 
				we can find a ``negative cycle'' as in the case of difference constraint (DC) systems.
				Here, the gainfree property assures that such a negative cycle, together with paths to the tails of hyperarcs in the cycle, corresponds to an infeasible subsystem of \eqref{eq:dual-LP}.
				The vector $\bs{r}^{(m)}_{v}$ stores how the negative cycle is derived from constraints in \eqref{eq:dual-LP} and helps to compute such a subsystem (with multiplicity) in \textsc{FarkasCertificateOfDualInfeasibility}.
				
				We first treat the case where the ``if '' condition in line 20 is false and the ``else if '' condition in line 22 is true in \textsc{DualFeasibility}.

				\begin{lemma}\label{lem:dual-infeasible-head-empty}
					If 
					\textsc{DualFeasibility} returns $\false$ 
					as the ``if '' condition in line 20 is false and the ``else if '' condition in line 22 is true, 
					then the dual LP problem~\eqref{eq:dual-LP} is infeasible and \textsc{FarkasCertificateOfDualInfeasibility} returns a Farkas' certificate of the dual infeasibility.
				\end{lemma}
				
				To show this lemma, we need some auxiliary claims.
				
				\begin{claim}\label{cl:nontriv=constant}
					In the end of \textsc{DualFeasibility}, 
					for all $k \in \{1,\dots,m\}$ and $v \in V$, 
					$y^{(k)}(v)$ contains $M$ if and only if 
					$\nontriv^{(k)}(v)=\false$.
					Moreover, if $y^{(k)}(v)$ contains $M$, the coefficient of $M$ is positive for all $k=1,\dots, m$ and $v \in V$.
				\end{claim}
				
				\begin{proof}
					We show these by induction on $k$ when the outer for-loop of \textsc{DualFeasibility} finishes the $k$th iteration.
					Let $k = 1$ and fix $v \in V$.
					Assume that $\nontriv^{(1)}(v)=\true$.
					Then, $\dual^{(1)}(v)$ is updated using $E \in \mathcal{E}$ with $T(E)=\emptyset$ by lines 10--12 in \textsc{DualFeasibility}.
					Hence, $\dual^{(1)}(v) = \ell(E)$ and thus it does not contain $M$.
					On the other hand, assume that $\nontriv^{(1)}(v)=\false$.
					If $\dual^{(1)}(v) = \dual^{(0)}(v) = M$, then it contains $M$ and its coefficient is one and thus positive.
					If $\dual^{(1)}(v) \neq \dual^{(0)}(v)$, then there exists $E \in \mathcal{E}$ such that $\dual^{(1)}(v)=\sum_{u \in T(E)}\weight(E,u)\dual^{(0)}(u) + \ell(E) = \sum_{u \in T(E)}\weight(E,u)M + \ell(E)$.
					Since $\nontriv^{(1)}(v)=\false$, we have $T(E) \neq \emptyset$.
					Therefore, $\dual^{(1)}(v)$ contains $M$ and its coefficient is positive since $\weight(E,u) > 0$ for each $u \in T(E)$.
					This completes the proof for $k=1$.
					
					For $k > 1$, fix $v \in V$.
					Assume that $\nontriv^{(k)}(v)=\true$.
					If $q(v) = k$, then $\dual^{(k)}(v)$ is updated using $E \in \mathcal{E}$ with $\nontriv^{(k-1)}(u)=\true$ for all $u \in T(E)$ by lines 10--12 in \textsc{DualFeasibility}.
					Then, by the inductive hypothesis, 
					$\dual^{(k-1)}(u)$ does not contain $M$ for all $u \in T(E)$.
					Hence, $\dual^{(k)}(v) = \ell(E) + \sum_{u \in T(E)}\weight(E,u)\dual^{(k-1)}(u)$ does not contain $M$.
					If $q(v) < k$, then $\dual^{(q(v))}(v)$ does not contain $M$ by inductive hypothesis.
					Since $\dual^{(k)}(v) \le \dual^{(q(v))}(v)$, 
					$\dual^{(k)}(v)$ does not contain $M$ either.
					On the other hand, assume that $\nontriv^{(k)}(v)=\false$.
					If $\dual^{(k)}(v) = \dual^{(k-1)}(v)$, then $\dual^{(k)}(v)$ contains $M$ and its coefficient is positive by the inductive hypothesis.
					If $\dual^{(k)}(v) \neq \dual^{(k-1)}(v)$, then there exists $E \in \mathcal{E}$ such that $\dual^{(k)}(v)=\ell(E) + \sum_{u \in T(E)}\weight(E,u)\dual^{(k-1)}(u)$.
					Since $\nontriv^{(k)}(v)=\false$, for some $u' \in T(E)$ we have $\nontriv^{(k-1)}(u')=\false$.
					By the inductive hypothesis, $\dual^{(k-1)}(u')$ contains $M$.
					Therefore, $\dual^{(k)}(v)$ also contains $M$ and its coefficient is positive since for each $u \in T(E)$ we have $\weight(E,u) > 0$ and the coefficient of $M$ in $\dual^{(k-1)}(u)$ with $\nontriv^{(k-1)}(u)=\false$ is positive by the inductive hypothesis.
					This completes the proof.
				\end{proof}
				
				\begin{claim}\label{cl:each-primal-ray-property}
					In the end of \textsc{DualFeasibility}, 
					for all $k \in \{1,\dots,m\}$ and $v \in V$, we have 
					$A\bs{\pray}^{(k)}_v \le  \bs{e}_v$, $\bs{\pray}^{(k)}_v \ge \bs{0}$, and 
					$\bs{c}^{T}\bs{\pray}^{(k)}_v$ equals the constant term of $\dual^{(k)}(v)$.
					If $\nontriv^{(k)}(v) = \true$, then 
					$A\bs{\pray}^{(k)}_v = \bs{e}_v$ and $\bs{c}^{T}\bs{\pray}^{(k)}_v = \dual^{(k)}(v)$.
				\end{claim}
				
				\begin{proof}
					We show this by induction on $k$.
					Assume that $k=1$.
					Fix $v \in V$.
					If $\change^{(1)}(v) = \true$, then $\bs{\pray}^{(1)}_{v} = e_{E} + \sum_{u\in T(E)}\weight(E,u)\bs{\pray}^{(0)}_{u}$ for some $E \in \mathcal{E}$ with $h(E) = v$.
					Since $\bs{\pray}^{(0)}_{u} = \bs{0}$ for all $u \in T(E)$, we have $\bs{\pray}^{(1)}_{v} = \bs{e}_E \ge \bs{0}$.
					Then, $A\bs{\pray}^{(1)}_v = A\bs{e}_E = \bs{e}_v - \sum_{u\in T(E)}\weight(E,u)\bs{e}_u \le \bs{e}_v$.
					Moreover, $\bs{c}^{T}\bs{\pray}^{(1)}_v = \bs{c}^{T}\bs{e}_E = \ell(E)$ and $\dual^{(1)}(v) = \ell(E) + \sum_{u \in T(E)}\weight(E,u)\dual^{(0)}(u)$.
					Since $\dual^{(0)}(u) = M$ for every $u \in T(E)$, 
					$\bs{c}^{T}\bs{\pray}^{(0)}_v$ equals the constant term of $\dual^{(0)}(v)$.
					Furthermore, if $\nontriv^{(1)}(v) = \true$, then $T(E) = \emptyset$ from lines 10--12 in \textsc{DualFeasibility}.
					Hence, $A\bs{\pray}^{(1)}_v = A\bs{e}_E = \bs{e}_v - \sum_{u\in T(E)}\weight(E,u)\bs{e}_u = \bs{e}_v$.
					Moreover, $\dual^{(1)}(v) = \ell(E) + \sum_{u \in T(E)}\weight(E,u)\dual^{(0)}(u) = \ell(E)$.
					Therefore, $\bs{c}^{T}\bs{\pray}^{(1)}_v = \ell(E) = \dual^{(1)}(v)$.
					If $\change^{(1)}(v) = \false$, then $\bs{\pray}^{(1)}_v = \bs{\pray}^{(0)}_v = 0$ and  $A\bs{\pray}^{(1)}_v = \bs{0}$. 
					Moreover, $\bs{c}^{T}\bs{\pray}^{(1)}_v = 0$ and $\dual^{(1)}(v) = \dual^{(0)}(v) = M$, and thus the constant term of $\dual^{(1)}(v)$ equals $\bs{c}^{T}\bs{\pray}^{(1)}_v$.
					Furthermore, $\nontriv^{(1)}(v) = \false$, and there exists nothing further to prove.
					
					Now, consider $k$ with $2 \le k \le m$.
					Fix $v \in V$.
					If $\change^{(k)}(v) = \false$, then the claimed properties hold by inductive hypothesis.
					If $\change^{(k)}(v) = \true$, then $\bs{\pray}^{(k)}_{v} = \bs{e}_E + \sum_{u \in T(E)}\weight(E,u)\bs{\pray}^{(k-1)}_{u}$ for some $E \in \mathcal{E}$ with $h(E) = v$. 
					Since $\weight(E,u) > 0$ by definition and $\bs{\pray}^{(k-1)}_{u} \ge 0$ by the inductive hypothesis for $u \in T(E)$, 
					we have $\bs{\pray}^{(k)}_{v} \ge 0$.
					Moreover, since $A\bs{\pray}^{(k-1)}_{u} \le \bs{e}_u$ $(u \in T(E))$ by the inductive hypothesis, we have 
					\begin{align*}
						A\bs{\pray}^{(k)}_{v} &= A\bs{e}_E + A\sum_{u \in T(E)}\weight(E,u)\bs{\pray}^{(k-1)}_{u}\\
						&= (\bs{e}_v - \sum_{u\in T(E)}\weight(E,u)\bs{e}_u) + \sum_{u\in T(E)}\weight(E,u)A\bs{\pray}^{(k-1)}_{u} \\
						&\le (\bs{e}_v - \sum_{u\in T(E)}\weight(E,u)\bs{e}_u) + \sum_{u\in T(E)}\weight(E,u)\bs{e}_u \\
						&= \bs{e}_v.
					\end{align*}
					Furthermore, $\bs{c}^{T}\bs{\pray}^{(k)}_v = \bs{c}^{T} \bs{e}_E + \sum_{u \in T(E)}\weight(E,u)\bs{c}^{T}\bs{\pray}^{(k-1)}_{u} = \ell(E) + \sum_{u \in T(E)}\weight(E,u)\bs{c}^{T}\bs{\pray}^{(k-1)}_{u}$, and $\dual^{(k)}(v) = \ell(E) + \sum_{u \in T(E)}\weight(E,u)\dual^{(k-1)}(u)$.
					By inductive hypothesis, for $u \in T(E)$, 
					$\bs{c}^{T}\bs{\pray}^{(k-1)}_{u}$ equals the constant term of $\dual^{(k-1)}(u)$.
					Hence, $\bs{c}^{T}\bs{\pray}^{(k)}_{v}$ equals the constant term of $\dual^{(k)}(v)$.
					
					Now, assume that $\nontriv^{(k)}(v) = \true$.
					Then from \cref{cl:nontriv=constant} $\dual^{(k)}(v)$ does not contain $M$.
					Moreover, $\nontriv^{(k-1)}(u) = \true$ for $u \in T(E)$, since otherwise for some $u \in T(E)$ $y^{(k-1)}(u)$ contains $M$ from \cref{cl:nontriv=constant} and so does $y^{(k)}(v)$, a contradiction.
					Therefore, we have $A\bs{\pray}^{(k-1)}_u = \bs{e}_u$ and $\bs{c}^{T}\bs{\pray}^{(k-1)}_u = \dual^{(k-1)}(u)$ for $u \in T(E)$ by inductive hypothesis.
					Hence, we have $A\bs{\pray}^{(k)}_{v} = A\bs{e}_E + A\sum_{u \in T(E)}\weight(E,u)\bs{\pray}^{(k-1)}_{u} = (\bs{e}_v - \sum_{u\in T(E)}\weight(E,u)\bs{e}_u) + \sum_{u\in T(E)}\weight(E,u)A\bs{\pray}^{(k-1)}_{u} = (\bs{e}_v - \sum_{u\in T(E)}\weight(E,u)\bs{e}_u) + \sum_{u\in T(E)}\weight(E,u)\bs{e}_u = \bs{e}_v$.
					Moreover, 
					$\bs{c}^{T}\bs{\pray}^{(k)}_{v} = \bs{c}^{T}\bs{e}_E + \sum_{u \in T(E)}\weight(E,u)\bs{c}^{T}\bs{\pray}^{(k-1)}_{u} = \ell(E) + \sum_{u \in T(E)}\weight(E,u)\dual^{(k-1)}(u)$ by the inductive hypothesis.
					Since $\dual^{(k)}(v) = \ell(E) + \sum_{u \in T(E)}\weight(E,u)\dual^{(k-1)}(u)$, we have $\bs{c}^{T}\bs{\pray}^{(k)}_{v} = \dual^{(k)}(v)$.
					This completes the proof.
				\end{proof}
				
				Now, we are ready to show \cref{lem:dual-infeasible-head-empty}.
				
				\begin{proof}[Proof of \cref{lem:dual-infeasible-head-empty}]
					We show that $\bs{\pray}^*$ is actually a Farkas' certificate of the dual infeasibility, i.e., (i) $\bs{\pray}^* \ge \bs{0}$, (ii) $A\bs{\pray}^* = \bs{0}$, and (iii) $\bs{c}^{T}\bs{\pray}^* < 0$ (see \cref{lem:Farkas-dual}).
					
					For (i), from \cref{cl:each-primal-ray-property}, we have that 
					$\bs{\pray}^* (=\bs{e}_E + \sum_{u \in T(E)}\weight(E,u)\bs{\pray}^{(m)}_{u})$ is a sum of nonnegative vectors.
					Hence, $\bs{\pray}^* \ge \bs{0}$.
					
					For (ii), observe that to satisfy $0 > \ell(E) + \sum_{u \in T(E)}\weight(E,u)\dual^{(m)}(u)$, $y^{(m)}(u)$ must not contain $M$ for each $u \in T(E)$, since otherwise the right-hand side of the inequality contains $M$ with a positive coefficient from \cref{cl:nontriv=constant} and thus greater than zero.
					Hence, for each $u \in T(E)$ $\nontriv^{(m)}(u) = \true$ from \cref{cl:nontriv=constant}, 
					implying that $A\bs{\pray}^{(m)}_u = \bs{e}_u$ from \cref{cl:each-primal-ray-property}.
					Therefore, we have 
					\begin{align*}
						A\bs{\pray}^* &=A\bs{e}_E + \sum_{u \in T(E)}\weight(E,u)A\bs{\pray}^{(m)}_{u}\\
						&= -\sum_{u \in T(E)}\weight(E,u)\bs{e}_u + \sum_{u \in T(E)}\weight(E,u)\bs{e}_u\\
						&=0.
					\end{align*}
					
					For (iii), for each $u \in T(E)$ we have $\bs{c}^{T}\bs{\pray}^{(m)}_u = \dual^{(m)}(u)$ from \cref{cl:each-primal-ray-property} since $\nontriv^{(m)}(u) = \true$.
					Hence, we have 
					\begin{align*}
						\bs{c}^{T}\bs{\pray}^* &=
						\bs{c}^{T}\bs{e}_E + \sum_{u \in T(E)}\weight(E,u)\bs{c}^{T}\bs{\pray}^{(m)}_{u}\\
						&=\ell(E) + \sum_{u \in T(E)}\weight(E,u)\dual^{(m)}(u)\\
						&< 0.
					\end{align*}
					
					Hence, $\bs{\pray}^*$ is a Farkas' certificate of the dual infeasibility and by \cref{lem:Farkas-dual} the dual LP problem~\eqref{eq:dual-LP} is infeasible.
				\end{proof}
				
				We then deal with the case where the ``if '' condition in line 20 is true in \textsc{DualFeasibility}.
				
				\begin{lemma}\label{lem:dual-infeasible-head-nonempty}
					If 
					\textsc{DualFeasibility} returns $\false$ 
					as the ``if '' condition in line 20 is true, 
					then the dual LP problem~\eqref{eq:dual-LP} is infeasible and \textsc{FarkasCertificateOfDualInfeasibility} returns a Farkas' certificate of the dual infeasibility.
				\end{lemma}
				
				To show \cref{lem:dual-infeasible-head-nonempty}, we need further auxiliary claims.
				
				\begin{claim}\label{cl:y-is-monotone-increasing}
					In the end of \textsc{DualFeasibility}, 
					we have $\dual^{(k)}(v) \le \dual^{(k-1)}(v)$ for all $k=1,\dots, m$ and $v \in V$.
					Moreover, $\dual^{(k)}(v) < \dual^{(k-1)}(v)$ if and only if $\change^{(k)}(v) = \true$ for all $k=1,\dots, m$ and $v \in V$.
					Furthermore, $\dual^{(k)}(v) \le \min\left\{ \ell(E) + \sum_{u \in T(E)}\weight(E,u)\dual^{(k-1)}(u) \mid E \in \mathcal{E}, h(E) = v \right\}$ for all $k=1,\dots, m$ and $v \in V$.
				\end{claim}
				
				\begin{proof}
					Fix $k = 1,\dots, m$ and $v \in V$.
					If $\change^{(k)}(v) = \false$, then $\dual^{(k)}(v) = \dual^{(k-1)}(v)$.
					If $\change^{(k)}=\true$, then $\dual^{(k-1)}(v) > \min\left\{ \ell(E) + \sum_{u \in T(E)}\weight(E,u)\dual^{(k-1)}(u) \mid E \in \mathcal{E}, h(E) = v \right\} = \dual^{(k)}(v)$.
					Therefore, $\dual^{(k)}(v) \le \dual^{(k-1)}(v)$ holds, and $\dual^{(k)}(v) < \dual^{(k-1)}(v)$ if and only if $\change^{(k)}(v) = \true$.
					Moreover, if $\change^{(k)}(v) = \false$, then $\dual^{(k)}(v) \le \min\left\{ \ell(E) + \sum_{u \in T(E)}\weight(E,u)\dual^{(k-1)}(u) \mid E \in \mathcal{E}, h(E) = v \right\}$.
					If $\change^{(k)}=\true$, then $\dual^{(k)}(v) = \min\left\{ \ell(E) + \sum_{u \in T(E)}\weight(E,u)\dual^{(k-1)}(u) \mid E \in \mathcal{E}, h(E) = v \right\}$.
					Summarizing the above, we have $\dual^{(k)}(v) \le \min\left\{ \ell(E) + \sum_{u \in T(E)}\weight(E,u)\dual^{(k-1)}(u) \mid E \in \mathcal{E}, h(E) = v \right\}$.
					This completes the proof.
				\end{proof}
				
				\begin{claim}\label{cl:extended-change-true-exists}
					In 
					\textsc{FarkasCertificateOfDualInfeasibility}, 
					for each $k=m+1,m, \dots, s+1$, there exists $u \in T(E^{(k)})$ such that $\change^{(k-1)}(u) = \true$.
				\end{claim}
				\begin{proof}
					We show this by induction on $k$ in the reverse order.
					Let $k=m+1$.
					Since we moved to procedure \textsc{FarkasCertificateOfDualInfeasibility}, 
					we have $\dual^{(m)}(w_{m+1}) > \ell(E^{(m+1)}) + \sum_{u \in T(E^{(m+1)})}\weight(E^{(m+1)},u)\dual^{(m)}(u)$.
					If $\change^{(m)}(u) = \false$ for all $u \in T(E^{(m+1)})$, then 
					$\dual^{(m)}(u) = \dual^{(m-1)}(u)$ for all $u \in T(E^{(m+1)})$.
					However, from \cref{cl:y-is-monotone-increasing} this implies $\dual^{(m)}(w_{m+1}) \le \ell(E^{(m+1)}) + \sum_{u \in T(E^{(m+1)})}\weight(E^{(m+1)},u)\dual^{(m-1)}(u) = \ell(E^{(m+1)}) + \sum_{u \in T(E^{(m+1)})}\weight(E^{(m+1)},u)\dual^{(m)}(u)$, 
					a contradiction.
					Hence, there exists $u \in T(E^{(m+1)})$ such that $\change^{(m)}(u) = \true$.
					
					For $k<m+1$, assume that there exists $u \in T(E_{k+1})$ such that $\change^{(k)}(u) = \true$ and chose $w_k \in T(E_{k+1})$ such that $\change^{(k)}(w_k) = \true$.
					This means that $\dual^{(k-1)}(w_k) > \ell(E_k) + \sum_{u \in T(E_k)}\weight(E_k,u)\dual^{(k-1)}(u)$ by Algorithm~\ref{alg:dual-feasibility}, since $E_k = p_k(w_k)$.
					If $\change^{(k-1)}(u) = \false$ for all $u \in T(E_k)$, then 
					$\dual^{(k-1)}(u) = \dual^{(k-2)}(u)$ for all $u \in T(E_k)$.
					However, from Claim~\ref{cl:y-is-monotone-increasing}, this implies $\dual^{(k-1)}(w_k) \le \ell(E_k) + \sum_{u \in T(E_k)}\weight(E_k,u)\dual^{(k-2)}(u) = \ell(E_k) + \sum_{u \in T(E_k)}\weight(E_k,u)\dual^{(k-1)}(u)$, 
					contradicting $\dual^{(k-1)}(w_k) > \ell(E_k) + \sum_{u \in T(E_k)}\weight(E_k,u)\dual^{(k-1)}(u)$.
					Hence, there exists $u \in T(E_k)$ such that $\change^{(k-1)}(u) = \true$.
					This completes the proof.
				\end{proof}
				
				\begin{claim}
					In 
					\textsc{FarkasCertificateOfDualInfeasibility}, 
					we can always obtain a cycle.
				\end{claim}
				\begin{proof}
					This follows from Claim~\ref{cl:extended-change-true-exists}.
					Indeed, in the for-loop in line 9 in \textsc{FarkasCertificateOfDualInfeasibility}, 
					we can always find $u \in T(E_k)$ with $\change^{(k-1)}(u) = \true$ for each $k=m+1,\dots, 1$ by Claim~\ref{cl:extended-change-true-exists}.
					Moreover, the number of vertices is $m$.
					Therefore, in line 12, the condition of the ``if'' must be true in $m+1$ loops.
					This means that we have found a cycle.
					This completes the proof.
				\end{proof}
				
				The following claim uses the gainfree property of the LP problem~\eqref{eq:LP}.
				
				\begin{claim}
					\label{cl:beard-nontriv}
					In 
					\textsc{FarkasCertificateOfDualInfeasibility}, 
					for any $s + 1 \le k \le t$ and any $u \in T(E^{(k)}) \setminus \{w_{k-1}\}$, 
					we have $\nontriv^{(k-1)}(u) = \true$.
				\end{claim}
				\begin{proof}
					We divide the proof into cases of $\nontriv^{(t)}(w_{t}) = \true$ and $\nontriv^{(t)}(w_{t}) = \false$.
					
					\paragraph*{Case 1: $\nontriv^{(t)}(w_{t}) = \true$.}
					We show more strongly that 
					for any $s + 1 \le k \le t$ and any $u \in T(E^{(k)})$, 
					we have $\nontriv^{(k-1)}(u) = \true$.
					We show this by induction on $k=t, t-1,\dots, s+1$.
					For $k=t$, since $\nontriv^{(t)}(w_{t}) = \true$, 
					$\dual^{(t)}(w_t)$ does not contain $M$ by \cref{cl:nontriv=constant}.
					As $\dual^{(t)}(w_t) = \sum_{u \in T(E^{(t)})}\weight(E^{(t)},u)\dual^{(t-1)}(u) + \ell(E^{(t)}))$, for every $u \in T(E^{(t)})$, $\dual^{(t-1)}(u)$ does not contain $M$.
					Hence for every $u \in T(E^{(t)})$, $\nontriv^{(t-1)}(u) = \true$ again by \cref{cl:nontriv=constant}.
					For $k < t$, we have $\nontriv^{(k)}(w_k) = \true$ by inductive hypothesis.
					Thus, $\dual^{(k)}(w_k)$ does not contain $M$, and for every $u \in T(E^{(k)})$ $\dual^{(k-1)}(u)$ does not contain $M$, implying that $\nontriv^{(k-1)}(u) = \true$ by \cref{cl:nontriv=constant}.
					
					\paragraph*{Case 2: $\nontriv^{(t)}(w_{t}) = \false$.}
					Recall that for each $s+1 \le k \le t$
					\begin{align*}
						\dual^{(k)}(w_k) &= \ell(E^{(k)}) + \sum_{u \in T(E^{(k)})}\weight(E^{(k)},u)\dual^{(k-1)}(u) \\
						&= \ell(E^{(k)}) + \sum_{u\in T(E^{(k)})\setminus \{w_{k-1}\}}\weight(E^{(k)},u)\dual^{(k-1)}(u) + \weight(E^{(k)},w_{k-1})\dual^{(k-1)}(w_{k-1}),
					\end{align*} since $\change^{(k)}(w_k)=\true$ and $w_{k-1} \in T(E^{(k)})$.
					Hence,  
					\begin{align}
						\begin{aligned}
							\label{eq:dual-beta-alpha}
							\dual^{(t)}(w_t) &= \ell(E^{(t)}) + \sum_{u\in T(E^{(t)})\setminus \{w_{t-1}\}}\weight(E^{(t)},u)\dual^{(t-1)}(u) + \weight(E^{(t)},w_{t-1})\dual^{(t-1)}(w_{t-1})\\
							&= \ell(E^{(t)}) + \sum_{u\in T(E^{(t)})\setminus \{w_{t-1}\}}\weight(E^{(t)},u)\dual^{(t-1)}(u) \\
							& \ \ \ +\weight(E^{(t)},w_{t-1})\left( \ell(E^{(t-1)}) + \sum_{u\in T(E^{(t-1)})\setminus \{w_{t-2}\}}\weight(E^{(t-1)},u)\dual^{(t-2)}(u) + \weight(E^{(t-1)},w_{t-2})\dual^{(t-2)}(w_{t-2}) \right)\\
							&=\cdots\\
							&=\sum_{k=s+1}^t \left(\prod_{\ell=k+1}^{t}\weight(E^{(\ell)},w_{\ell-1})\right)\left(\ell(E^{(k)}) + \sum_{u\in T(E^{(k)})\setminus \{w_{k-1}\}}\weight(E^{(k)},u)\dual^{(k-1)}(u)\right)\\
							& \ \ \ + \left(\prod_{k=s+1}^{t}\weight(E^{(k)},w_{k-1})\right) \dual^{(s)}(w_s).
						\end{aligned}
					\end{align}
					Let $\dual^{(t)}(w_t) = dM+f$ and $\dual^{(s)}(w_s) = d'M+f'$.
					Since $\dual^{(t)}(w_t) < \dual^{(s)}(w_s)$, 
					we have $d < d'$ or ($d=d'$ and $f < f'$), and in particular $d \le d'$.
					From \cref{eq:dual-beta-alpha}, 
					we have $d \ge \left(\prod_{k=s+1}^{t}\weight(E^{(k)},w_{k-1})\right)d'$ as the coefficient of $M$ in any $\dual^{(k-1)}(u)$ is nonnegative from \cref{cl:nontriv=constant}.
					Now, $\prod_{k=s+1}^{t}\weight(E^{(k)},w_{k-1}) \ge 1$, since the LP problem is gainfree.
					Since $d \le d'$, 
					it follows that $\prod_{k=s+1}^{t}\weight(E^{(k)},w_{k-1}) = 1$ and 
					$d=d'$.
					Moreover, from \cref{eq:dual-beta-alpha} for any $s + 1 \le k \le t$ and any $u \in T(E^{(k)}) \setminus \{w_{k-1}\}$, $\dual^{(k-1)}(u)$ does not contain $M$, since otherwise $d > d'$ by \cref{cl:nontriv=constant}, a contradiction.
					Therefore, for any $s + 1 \le k \le t$ and any $u \in T(E^{(k)}) \setminus \{w_{k-1}\}$, we have $\nontriv^{(k-1)}(u) = \true$ again from \cref{cl:nontriv=constant}.
					This completes the proof.
				\end{proof}
				
				Now, we are ready to prove \cref{lem:dual-infeasible-head-nonempty}.
				
				\begin{proof}[Proof of \cref{lem:dual-infeasible-head-nonempty}]
					We show that $\bs{\pray}^*$ is actually a Farkas' certificate of the dual infeasibility, i.e., (i) $\bs{\pray}^* \ge \bs{0}$, (ii) $A\bs{\pray}^* = \bs{0}$, and (iii) $\bs{c}^{T}\bs{\pray}^* < 0$.
					
					For (i), 
					recall that for each $s+1 \le k \le t$,
					\begin{align*}
						\bs{\pray}^{(k)}_{w_k} = e_{E^{(k)}} + \sum_{u\in T(E^{(k)})}\weight(E^{(k)},u)\bs{\pray}^{(k-1)}_{u} = e_{E^{(k)}} + \sum_{u\in T(E^{(k)})\setminus \{w_{k-1}\}}\weight(E^{(k)},u)\bs{\pray}^{(k-1)}_{u} + \weight(E^{(k)},w_{k-1})\bs{\pray}^{(k-1)}_{w_{k-1}},
					\end{align*} since $\change^{(k)}(w_k)=\true$ and $w_{k-1} \in T(E^{(k)})$.
					Hence,  
					\begin{align*}
						\bs{\pray}^{(t)}_{w_t} &= \bs{e}_{E^{(t)}} + \sum_{u \in T(E^{(t)})\setminus \{w_{t-1}\}} \weight(E^{(t)},u)\bs{\pray}^{(t-1)}_{u} + \weight(E^{(t)},w_{t-1})\bs{\pray}^{(t-1)}_{w_{t-1}}\\
						&= \bs{e}_{E^{(t)}} + \sum_{u \in T(E^{(t)})\setminus \{w_{t-1}\}} \weight(E^{(t)},u)\bs{\pray}^{(t-1)}_{u} \\
						& \ \ \ +\weight(E^{(t)},w_{t-1})\left( \bs{e}_{E^{(t-1)}} + \sum_{u \in T(E^{(t-1)})\setminus \{w_{t-2}\}} \weight(E_{t-1},u)\bs{\pray}^{(t-2)}_{u} + \weight(E^{(t-1)},w_{t-2})\bs{\pray}^{(t-2)}_{w_{t-2}} \right)\\
						&=\cdots\\
						&=\sum_{k=s+1}^t\prod_{\ell=k+1}^{t}\weight(E^{(\ell)},w_{\ell-1})\left(e_{E^{(k)}} + \sum_{u\in T(E^{(k)})\setminus \{w_{k-1}\}}\weight(E^{(k)},u)\bs{\pray}^{(k-1)}_{u}\right)
						+ \prod_{k=s+1}^{t}\weight(E^{(k)},w_{k-1}) \bs{\pray}^{(s)}_{w_s}.
					\end{align*}
					Since the LP problem is gainfree, we have $\prod_{k=s+1}^{t}\weight(E^{(k)},w_{k-1}) \ge 1$.
					Hence, 
					\begin{align*}
						\bs{\pray}^* &= \bs{\pray}^{(t)}_{w_t} - \bs{\pray}^{(s)}_{w_s} \\
						&= \sum_{k=s+1}^t\prod_{\ell=k+1}^{t}\weight(E^{(\ell)},w_{\ell-1})\left(e_{E^{(k)}} + \sum_{u\in T(E^{(k)})\setminus \{w_{k-1}\}}\weight(E^{(k)},u)\bs{\pray}^{(k-1)}_{u}\right)
						+ \left(\prod_{k=s+1}^{t}\weight(E^{(k)},w_{k-1}) -1 \right)\bs{\pray}^{(s)}_{w_s}\\
						&\ge \bs{0},
					\end{align*}
					where we recall that $\bs{e}_{E^{(k)}} \ge 0$ and  $\bs{\pray}^{(k-1)}_{u} \ge 0$ (from \cref{cl:each-primal-ray-property}) 
					for each $k=s+1,\dots, t$ and $u \in T(E^{(k)})$.
					Thus, $\bs{\pray}^* \ge \bs{0}$.
					
					For (ii), recall that 
					for any $s + 1 \le k \le t$ and any $u \in T(E^{(k)}) \setminus \{w_{k-1}\}$, 
					we have $A\bs{\pray}^{(k-1)}_{u} = \bs{e}_{u}$ from \cref{cl:each-primal-ray-property,cl:beard-nontriv}.
					Moreover, we have 
					$A\bs{e}_{E^{(k)}} = \bs{e}_{h(E^{(k)})}-\sum_{u \in T(E^{(k)})} \weight(E^{(k)},u)\bs{e}_u$.
					Hence, for each $s +1 \le k \le t$,
					\begin{align*}
						A\bs{\pray}^{(k)}_{w_k} &=A(e_{E^{(k)}} + \sum_{u\in T(E^{(k)})}\weight(E^{(k)},u)\bs{\pray}^{(k-1)}_{u})\\
						&=\bs{e}_{h(E^{(k)})}-\sum_{u \in T(E^{(k)})} \weight(E^{(k)},u)\bs{e}_u + A(\sum_{u\in T(E^{(k)})}\weight(E^{(k)},u)\bs{\pray}^{(k-1)}_{u})\\
						&=\bs{e}_{w_{k}} + \sum_{u\in T(E^{(k)})}\weight(E^{(k)},u)(A\bs{\pray}^{(k-1)}_{u} - \bs{e}_u)\\
						&=\bs{e}_{w_{k}} + \sum_{u\in T(E^{(k)})\setminus \{w_{k-1}\}}\weight(E^{(k)},u)(A\bs{\pray}^{(k-1)}_{u} - \bs{e}_u) + \weight(E^{(k)},w_{k-1})(A\bs{\pray}^{(w_{k-1})}_{k-1}-\bs{e}_{w_k-1})\\
						&= \bs{e}_{w_{k}} + \weight(E^{(k)},w_{k-1})(A\bs{\pray}^{(w_{k-1})}_{k-1}-\bs{e}_{w_k-1}).
					\end{align*}
					Namely, we have $A\bs{\pray}^{(k)}_{w_k} - \bs{e}_{w_{k}} = \weight(E^{(k)},w_{k-1})(A\bs{\pray}^{(w_{k-1})}_{k-1}-\bs{e}_{w_k-1})$.
					Therefore, we have
					\begin{align*}
						A\bs{\pray}^{(t)}_{w_t}  - \bs{e}_{w_{t}}&= \weight(E^{(t)},w_{t-1})(A\bs{\pray}^{(t-1)}_{w_{t-1}}-\bs{e}_{w_t-1})\\
						&= \weight(E^{(t)},w_{t-1})\weight(E^{(t-1)},w_{t-2})(A\bs{\pray}^{(t-2)}_{w_{t-2}}-\bs{e}_{w_t-2})\\
						&= \cdots \\
						&= \prod_{k=s+1}^t \weight(E^{(k)},w_{k-1}) (A\bs{\pray}^{(s)}_{w_s} - \bs{e}_{w_{s}}).
					\end{align*}
					Hence, we have 
					\begin{align*}
						A\bs{\pray}^* &=A(\bs{\pray}^{(t)}_{w_t} - \bs{\pray}^{(s)}_{w_s})\\
						&= \bs{e}_{w_{t}} + \prod_{k=s+1}^t \weight(E^{(k)},w_{k-1}) (A\bs{\pray}^{(s)}_{w_s}  - \bs{e}_{w_{s}})- A\bs{\pray}^{(s)}_{w_s}\\
						&= \bs{e}_{w_{t}} + \left(\prod_{k=s+1}^t \weight(E^{(k)},w_{k-1})-1\right) A\bs{\pray}^{(s)}_{w_s} - \prod_{k=s+1}^t \weight(E^{(k)},w_{k-1})\bs{e}_{w_{s}}.
					\end{align*}
					Now, if $\nontriv^{(t)}(w_{t}) = \true$, 
					then $\nontriv^{(s)}(w_{s}) = \true$ by the proof of \cref{cl:beard-nontriv}.
					Hence, $A\bs{\pray}^{(s)}_{w_s} = \bs{e}_{w_{s}}$ by \cref{cl:each-primal-ray-property}.
					Therefore, we have 
					\begin{align*}
						&\bs{e}_{w_{t}} + \left(\prod_{k=s+1}^t \weight(E^{(k)},w_{k-1})-1\right) A\bs{\pray}^{(s)}_{w_s} - \prod_{k=s+1}^t \weight(E^{(k)},w_{k-1})\bs{e}_{w_{s}} \\ 
						&= \bs{e}_{w_{t}} + (\prod_{k=s+1}^t \weight(E^{(k)},w_{k-1})-1) \bs{e}_{w_{s}} - \prod_{k=s+1}^t \weight(E^{(k)},w_{k-1})\bs{e}_{w_{s}}\\
						&= \bs{e}_{w_{t}} - \bs{e}_{w_{s}}\\
						&= \bs{0}, 
					\end{align*}
					where the last equality holds since $w_{t} = w_{s}$.
					If $\nontriv^{(t)}(w_{t}) = \false$, then 
					$\prod_{k=s+1}^t \weight(E^{(k)},w_{k-1}) = 1$ by the proof of \cref{cl:beard-nontriv}.
					Therefore, we have 
					\begin{align*}
						&\bs{e}_{w_{t}} + \left(\prod_{k=s+1}^t \weight(E^{(k)},w_{k-1})-1\right) A\bs{\pray}^{(s)}_{w_s} - \prod_{k=s+1}^t \weight(E^{(k)},w_{k-1})\bs{e}_{w_{s}} \\ 
						&= \bs{e}_{w_{t}} - \bs{e}_{w_{s}}\\
						&= \bs{0}.
					\end{align*}
					In either case, we have $A\bs{\pray}^* = \bs{0}$.
					
					For (iii), 
					if $\nontriv^{(t)}(w_{t}) = \true$, 
					then $\nontriv^{(s)}(w_{s}) = \true$ by the proof of \cref{cl:beard-nontriv}.
					Hence, $\bs{c}^{T}\bs{\pray}^{(k)}_{w_k} = \dual^{(k)}(w_k)$ for $k \in \{s,t\}$ by \cref{cl:each-primal-ray-property}.
					Hence, we have 
					\begin{align*}
						\bs{c}^{T}\bs{\pray}^* &=
						\bs{c}^{T}(\bs{\pray}^{(t)}_{w_t} - \bs{\pray}^{(s)}_{w_s})\\
						&=\dual^{(t)}(w_t)-\dual^{(s)}(w_s) \\
						& < 0,
					\end{align*}
					where the last strict inequality holds by Claim~\ref{cl:y-is-monotone-increasing} and $\change^{(t)}(w_t)=\true$ (by Algorithm~\ref{alg:primal-ray}).
					If $\nontriv^{(t)}(w_{t}) = \false$, $\bs{c}^{T}\bs{\pray}^{(k)}_{w_k}$ equals the constant term of $\dual^{(k)}(w_k)$ for $k \in \{s,t\}$ by \cref{cl:each-primal-ray-property}.
					As $\dual^{(t)}(w_t) < \dual^{(s)}(w_s)$ and the coefficients of $M$ in $\dual^{(t)}(w_t)$ and $\dual^{(s)}(w_s)$ coincide by the proof of \cref{cl:beard-nontriv}, the constant term of $\dual^{(t)}(w_t)$ is smaller than that of $\dual^{(s)}(w_s)$.
					Hence, we have 
					\begin{align*}
						\bs{c}^{T}\bs{\pray}^* &=
						\bs{c}^{T}(\bs{\pray}^{(t)}_{w_t} - \bs{\pray}^{(s)}_{w_s})\\
						&={\rm the\ constant\ term\ of\ }\dual^{(t)}(w_t)-{\rm the\ constant\ term\ of\ }\dual^{(s)}(w_s) \\
						& < 0.
					\end{align*}
					
					Hence, $\bs{\pray}^*$ is a Farkas' certificate of the dual infeasibility and by \cref{lem:Farkas-dual} the dual LP problem~\eqref{eq:dual-LP} is infeasible.
					This completes the proof.
				\end{proof}
				
				Combining \cref{lem:dual-infeasible-head-empty} and \cref{lem:dual-infeasible-head-nonempty}, we obtain \cref{lem:dual-infeasible}.
				
				Now, we are ready to show \cref{prop:dual-feasibility-correct}, which we recall:
				
				\newtheorem*{MainProp}{Proposition~\ref{prop:dual-feasibility-correct}}
				\begin{MainProp}
					\cref{alg:cert-gainfree-Leontief-dual} is a combinatorial 
					${\rm O}(m^3n)$-time 
					certifying algorithm for the feasibility of the dual~\eqref{eq:dual-LP} of the LP problem with a gainfree Leontief substitution system.
				\end{MainProp}
				
				\begin{proof}[Proof of \cref{prop:dual-feasibility-correct}]
					Note that subroutines \textsc{DualFeasibility}, \textsc{DualSolution}, and \textsc{FarkasCertificateOfDualInfeasibility} constitute a certifying algorithm for the feasibility problem of the dual LP problem~\eqref{eq:dual-LP} (\cref{alg:cert-gainfree-Leontief-dual}).
					The correctness of this algorithm follows from \cref{lem:dual-feasible,lem:dual-infeasible-head-empty,lem:dual-infeasible-head-nonempty}.
					
					Now, we analyze the running time of the algorithm.
					The most time-consuming part of the algorithm is the for-loop from line 2 to 19 in \textsc{DualFeasibility}.
					This for-loop has $m$ iterations, and 
					${\rm O}(mn)$ operations for computing $\bs{\pray}^{(k)}_v$ each $v \in V$ in each iteration.
					Hence, it takes ${\rm O}(m^3n)$ time.
					This completes the proof.
				\end{proof}
				
				\begin{remark}
					In the case of a DC system, 
					$\bs{\pray}^*$ in \textsc{FarkasCertificateOfDualInfeasibility} corresponds to a negative cycle.
					Namely, $\bs{\pray}^* \in \{ 0,1 \}^m$, and the arc set $B := \{ E \in \mathcal{E} \mid \bs{\pray}^*(E) = 1 \}$ constitutes a cycle, whose weight is negative.
					Hence, our algorithm is an extension of the Bellman-Ford algorithm.
					
					Note that the main differences from the Bellman-Ford algorithm is that our algorithm keeps values of the primal vectors (i.e., $\bs{\pray}^{(k)}_{v}$), which makes the running time of our algorithm slower than that of the Bellman-Ford algorithm.
					However, for DC systems we only need ${\rm O}(n)$ operations to compute $\bs{\pray}^{(k)}_v$ and thus our algorithm runs in ${\rm O}(m^2n)$ time, 
					which is ${\rm O}(m)$ times slower than 
					the running time ${\rm O}(mn)$ of the Bellman--Ford algorithm.
				\end{remark}

				\subsection{A certifying algorithm for the feasibility of the primal LP problem}
				\label{subsec:primal-feasibility}
				
				In this subsection, we provide a certifying algorithm for the feasibility of the primal LP problem~\eqref{eq:LP} with a gainfree Leontief substitution system, using the data computed in \textsc{DualFeasibility}.
				More precisely, we show that subroutines \textsc{PrimalFeasibility} (\cref{alg:primal-feasibility}), \textsc{PrimalSolution} (\cref{alg:primal-solution}), and \textsc{FarkasCertificateOfPrimalInfeasibility} (\cref{alg:dual-ray}), together with \textsc{DualFeasibility}, constitute a certifying algorithm for the feasibility problem of the primal LP problem~\eqref{eq:LP} (\cref{alg:cert-gainfree-Leontief-primal}).
				\textsc{PrimalFeasibility} determines the feasibility of the primal LP problem~\eqref{eq:LP} using the same criterion as in (ii) of Theorem 3.6 in \cite{JMR92}.
				\textsc{PrimalSolution} is similar to \textsc{PrimalRetrieval} in \cite{JMR92};
				however, \textsc{PrimalSolution} also computes a primal feasible solution when the dual LP problem is infeasible.
				\textsc{FarkasCertificateOfPrimalInfeasibility} returns a Farkas' certificate of the primal infeasibility, where the gainfree property is again crucial for the correctness.
							
				\begin{algorithm}
					\caption{Combinatorial certifying algorithm for the feasibility of the primal LP problems with gainfree Leontief substitution systems}
					\label{alg:cert-gainfree-Leontief-primal}
					\KwInput{A matrix $A$ and vectors $\bs{b}$ and $\bs{c}$ for the primal LP problem~\eqref{eq:LP}.}
					($\bs{y}^{(m)},\bs{r}^{(m)},\change^{(k)}(k=0,...,m),p^{(k)}(k=0,...,m),\nontriv^{(m)},\bs{q},\val$)$\ot$\textsc{DualFeasibility}($A,\bs{c}$). \\
					\eIf{\textsc{PrimalFeasibility}{\rm (}$\bs{b},\nontriv^{(m)}${\rm )} $=\true$ }{
						$\bs{\primal}^*\ot$ {\rm \textsc{PrimalSolution}}($A,\bs{b},\nontriv^{(m)},p^{(k)}(k=0,...,m),\bs{q},\val$).\\
						{\bf print} ``primal-feasible'' and 
						{\bf return} $\bs{\primal}^*$.}
					{$\bs{\dray}^*\ot {\rm \textsc{FarkasCertificateOfPrimalInfeasibility}}$($\bs{\dual}^{(m)},\nontriv^{(m)}$).\\
						{\bf print} ``primal-infeasible'' and  
						{\bf return} $\bs{\dray}^*$.}
				\end{algorithm}

				\begin{algorithm}
					\caption{\textsc{PrimalFeasibility}}
					\label{alg:primal-feasibility}
					\KwInput{A vector $\bs{b}$ and $\nontriv^{(m)}$.}
					\eIf{$b(v) = 0$ for all $v$ with $\nontriv^{(m)}(v)=\false$}{
						{\bf return} $\true$.
					}
					{
						{\bf return} $\false$.}
				\end{algorithm}	
								
				\begin{algorithm}
					\caption{\textsc{PrimalSolution} (\textsc{PrimalRetrieval} in~\cite{JMR92})}
					\label{alg:primal-solution}
					\KwInput{A matrix $A$ and a vector $\bs{b}$ for the constraint of the primal LP problem~\eqref{eq:LP}, and $\nontriv^{(m)},p^{(k)}(k=0,...,m),\bs{q},\val$.}
					\If{
						$\val = \false$}{($\bs{y}^{(m)},\bs{r}^{(m)},\change^{(k)}(k=0,...,m),p^{(k)}(k=0,...,m),\bs{q},\val$)$\ot$\textsc{DualFeasibility}($A,\bs{0}$).}
					For each $E \in \mathcal{E}$, $\primal^*(E)\ot 0$, $\Tilde{V}\ot \{i \in V \mid \nontriv^{(m)}(i) = \true\}$, and for each $v \in V$, $f(v)\ot b(v)$.\\
					\While{$\Tilde{V} \neq \emptyset$}{
						Choose an arbitrary $v \in \Tilde{V}$ with maximum $q(v)$.\\
						$E \ot p^{(q(v))}(v)$.\\
						$\primal^*(E)\ot f(v)$.\\
						$f(u)\ot f(u)+\weight(E,u)\primal(E)$ for each $u \in T(E)$.\\
						$\Tilde{V}\ot \Tilde{V} \setminus \{v\}$.
					}
					{\bf return} $\bs{\primal}^*$.
				\end{algorithm}
				
				\begin{algorithm}
					\caption{\textsc{FarkasCertificateOfPrimalInfeasibility}}
					\label{alg:dual-ray}
					\KwInput{A vector $\bs{\dual}^{(m)}$ and $\nontriv^{(m)}$.}
					\For{each $v \in V$}
					{\eIf{$\nontriv^{(m)}(v) = \true$}
						{$\dray^*(v) \ot 0$.}
						{$\dray^*(v) \ot $ the coefficient of $M$ in $\dual^{(m)}(v)$.}
					}
					{\bf return} $\bs{\dray}^*$.
				\end{algorithm}	
				
				Before going into the proofs of correctness of these algorithms, we show several examples how these algorithms work, continuing examples shown in \cref{subsec:dual-feasibility}.
				
				\begin{example}
					Recall \cref{ex:main-algo1} in \cref{subsec:dual-feasibility}.
					In this example, we have $\nontriv^{(m)}(v) = \true$ for each $v \in V$ and \textsc{PrimalFeasibility}$(\bs{b},\nontriv^{(m)})=\true$ for any $\bs{b} (\ge \bs{0})$.
					Hence, \textsc{PrimalSolution} is called in \cref{alg:cert-gainfree-Leontief-primal}.
					As the dual LP problem is infeasible, 
					\textsc{DualFeasibility}$(A,\bs{0})$ is called in \textsc{PrimalSolution} and in particular $\bs{q} = (1,1,1,1)^T$ is obtained.
					Then in the while-loop in \textsc{PrimalSolution} variables $\bs{\primal}^*$ and $\bs{f}$ are updated as follows.
					Initially, $\bs{\primal}^*=\bs{0}$ and $\bs{f} = (b_1,b_2,b_3,b_4)^T$.
					First, we may choose $v_1$ according to $\bs{q}$ and since $p^{(1)}(v_1) = E_4$, $\primal^*(E_4) = b_1$ and $\bs{f}$ remains unchanged.
					Then we may choose $v_2$ and since $p^{(1)}(v_2) = E_5$, $\primal^*(E_5) = b_2$ and $\bs{f}$ remains unchanged.
					Then we may choose $v_3$ and since $p^{(1)}(v_3) = E_6$, $\primal^*(E_6) = b_3$ and $\bs{f}$ remains unchanged.
					Finally, we choose $v_4$ and since $p^{(1)}(v_4) = E_7$, $\primal^*(E_7) = b_4$.
					Then we obtain a feasible solution $\bs{\primal}^*=(0,0,0,b_1,b_2,b_3,b_4)^T$ of the primal LP problem~\eqref{eq:LP}.
				\end{example}
				
				\begin{example}
					Recall \cref{ex:main-algo2} in \cref{subsec:dual-feasibility}.
					In this example, we have $\nontriv^{(m)}(v) = \true$ for each $v \in V$ and \textsc{PrimalFeasibility}$(\bs{b},\nontriv^{(m)})=\true$ for any $\bs{b} (\ge \bs{0})$.
					Hence, \textsc{PrimalSolution} is called in \cref{alg:cert-gainfree-Leontief-primal}.
					As the dual LP problem is infeasible, 
					\textsc{DualFeasibility}$(A,\bs{0})$ is called in \textsc{PrimalSolution} and in particular $\bs{q} = (3,2,1)^T$ is obtained.
					Then in the while-loop in \textsc{PrimalSolution} variables $\bs{\primal}^*$ and $\bs{f}$ are updated as follows.
					Initially, $\bs{\primal}^*=(0,0,0,0)^T$ and $\bs{f} = (b_1,b_2,b_3)^T$.
					Then first $v_1$ is chosen according to $\bs{q}$ and since $p^{(3)}(v_1) = E_2$, $\primal^*(E_2) = b_1$ and $\bs{f} = (b_1,2b_1+b_2,b_1+b_3)^T$.
					Then second $v_2$ is chosen and since $p^{(2)}(v_2) = E_3$, $\primal^*(E_3) = 2b_1+b_2$ and $\bs{f} = (b_1,2b_1+b_2,5b_1+2b_2+b_3)^T$.
					Finally, $v_3$ is chosen and since $p^{(1)}(v_3) = E_4$, $\primal^*(E_4) = 5b_1+2b_2+b_3$.
					Then we obtain a feasible solution $\bs{\primal}^*=(0,b_1,2b_1+b_2,5b_1+2b_2+b_3)^T$ of the primal LP problem~\eqref{eq:LP}.
				\end{example}
				
				\begin{example}
					Recall \cref{ex:main-algo3} in \cref{subsec:dual-feasibility}.
					In this example, we have $\nontriv^{(m)}(v) = \false$ for each $v \in V$ and \textsc{PrimalFeasibility}$(\bs{b},\nontriv^{(m)})=\true$ if and only if $\bs{b} =\bs{0}$.
					If $\bs{b} =\bs{0}$, then $\bs{\primal}^*=\bs{0}$ is trivially a feasible solution of the primal LP problem~\eqref{eq:LP}, which can be obtained by \textsc{PrimalSolution}.
					Assume that $\bs{b} \neq \bs{0}$.
					Then \textsc{FarkasCertificateOfPrimalInfeasibility} is called in \cref{alg:cert-gainfree-Leontief-primal} and since $\bs{\dual}^{(3)} = (M-2,(1/2)M-3,(1/6)M)$, $\bs{\dray}^* = (1,1/2,1/6)$ is obtained.
					Then $(\bs{\dray}^*)^TA=(0,0,0)^T$ and $(\bs{\dray}^*)^T\bs{b}=b_1+(1/2)b_2+(1/6)b_3 >0$.
					Hence, $\bs{\dray}^*$ is a Farkas certificate of infeasibility of the primal LP problem~\eqref{eq:LP}.
				\end{example}
				
				\begin{example}
					\cref{ex:main-algo4} in \cref{subsec:dual-feasibility} can be treated similarly to \cref{ex:main-algo3} in \cref{subsec:dual-feasibility} and we omit this case.
				\end{example}
				
				Now, we show the correctness of subroutines \textsc{PrimalFeasibility}, \textsc{PrimalSolution}, and \textsc{FarkasCertificateOfPrimalInfeasibility}, and show the following proposition.
				
				\begin{proposition}
					\label{prop:primal-feasibility-correct}
					\cref{alg:cert-gainfree-Leontief-primal} is a combinatorial 
					${\rm O}(m^3n)$-time 
					certifying algorithm for the feasibility of the primal LP problem~\eqref{eq:LP} with a gainfree Leontief substitution system.
				\end{proposition}
				
				To show \cref{prop:primal-feasibility-correct}, 
				it suffices to show that if the primal LP problem is feasible, then \textsc{PrimalSolution} returns a primal feasible solution, and 
				otherwise \textsc{FarkasCertificateOfPrimalInfeasibility} returns a Farkas' certificate of the primal infeasibility.
				
				We first consider the case where the primal LP problem is feasible.
				
				\begin{lemma}\label{lem:primal-feasible}
					If 
					\textsc{PrimalFeasibility} returns $\true$, 
					then the primal LP problem~\eqref{eq:LP} is feasible and \textsc{PrimalSolution} returns a feasible solution of \eqref{eq:LP}.
				\end{lemma}
				
				\begin{proof}
					We divide into two cases where the dual LP problem is feasible and infeasible.
					In the case where the dual LP problem is feasible, we can use Theorem 3.6 in \cite{JMR92} as follows.
					First, observe that \textsc{PrimalSolution} is essentially the same as the procedure \textsc{PrimalRetrieval} in \cite{JMR92}, where we note that $\nontriv^{(m)}$ coincides with $\nontriv$ in \cite{JMR92} although we introduced a symbol $M$.
					Then, from Theorem 3.6 in \cite{JMR92}, we know that the primal LP problem is feasible if and only if $b(v) = 0$ for all $v \in V$ with $\nontriv^{(m)}(v) = \false$, and 
					if it is feasible, then \textsc{PrimalSolution} outputs a feasible solution.
					In the case where the dual LP problem is infeasible, we set $\bs{c}=\bs{0}$ and run \textsc{DualFeasibility}, where the dual LP problem is always feasible.
					Then we can again use Theorem 3.6 in \cite{JMR92} and conclude that the primal LP problem is feasible if and only if $b(v) = 0$ for all $v \in V$ with $\nontriv^{(m)}(v) = \false$, where we note that $\nontriv^{(m)}$ depends only on $A$ (and independent from $\bs{c}$).    
				\end{proof}
				
				Now, we show that when $b(v) > 0$ for some $v \in V$ with $\nontriv^{(m)}(v) = \false$, 
				the vector $\bs{\dray}^*$ returned by \textsc{FarkasCertificateOfPrimalInfeasibility} is a Farkas' certificate of the primal infeasibility.
				Note that this holds in both cases where the dual LP problem is feasible and infeasible.
				We need an auxiliary claim, which is an extension of Lemma 3.4 (ii) in \cite{JMR92}.
				Let $V^{(k)}_{\nontriv} =\{ v \in V \mid \nontriv^{(k)}(v) = \true \}$.
				
				\begin{claim}\label{cl:nontrivial-level}
					If $V^{(k)}_{\nontriv} = V^{(k+1)}_{\nontriv}$ for some $0 \le k \le m-1$, 
					then $V^{(k)}_{\nontriv} = V^{(\ell)}_{\nontriv}$ for every $\ell \ge k+1$.
					It follows that for each $1 \le k \le m$ if $|V^{(k)}_{\nontriv}| \le k-1$, 
					then $V^{(k-1)}_{\nontriv} = V^{(\ell)}_{\nontriv}$ for every $\ell \ge k$.
				\end{claim}
				\begin{proof}
					Assume that $V^{(k)}_{\nontriv} = V^{(k+1)}_{\nontriv}$ for some $0 \le k \le m-1$.
					This means that for each $v \not\in V^{(k)}_{\nontriv}$ and each $E\in \mathcal{E}$ with $h(E)=v$, 
					we have $(V \setminus V^{(k)}_{\nontriv}) \cap T(E) \neq \emptyset$ as otherwise $v$ would be included in $V^{(k+1)}_{\nontriv}$.
					Since $V^{(k)}_{\nontriv} = V^{(k+1)}_{\nontriv}$, for each $v \not\in V^{(k+1)}_{\nontriv}$ and each $E\in \mathcal{E}$ with $h(E)=v$, 
					we have $(V \setminus V^{(k+1)}_{\nontriv}) \cap T(E) \neq \emptyset$ and thus $v \not\in V^{(k+2)}_{\nontriv}$.
					Hence, $V^{(k+2)}_{\nontriv} \subseteq V^{(k+1)}_{\nontriv}$.
					From definition, we also have $V^{(k+1)}_{\nontriv} \subseteq V^{(k+2)}_{\nontriv}$.
					Therefore, we have $V^{(k+1)}_{\nontriv} = V^{(k+2)}_{\nontriv}$.
					In a similar way, we can inductively show that $V^{(k)}_{\nontriv} = V^{(\ell)}_{\nontriv}$ for every $\ell \ge k+1$.
					
					Now we show the latter statement of the lemma.
					It suffices to show that $V^{(k-1)}_{\nontriv} = V^{(k)}_{\nontriv}$ from the former statement.
					Assume otherwise that $V^{(k-1)}_{\nontriv} \subsetneq V^{(k)}_{\nontriv}$.
					Then, from the former statement, $V^{(0)}_{\nontriv} \subsetneq \dots \subsetneq V^{(k-1)}_{\nontriv} \subsetneq V^{(k)}_{\nontriv}$ and thus 
					$|V^{(k)}_{\nontriv}| \ge k$, a contradiction.
					Hence, the latter statement holds.
				\end{proof}
				
				Now, we are ready to show that when $b(v) > 0$ for some $v \in V$ with $\nontriv^{(m)}(v) = \false$, 
				the vector $\bs{\dray}^*$ returned by \textsc{FarkasCertificateOfPrimalInfeasibility} is a Farkas' certificate of the primal infeasibility, i.e., 
				$(\bs{\dray}^*)^TA \le \bs{0}$ and $(\bs{\dray}^*)^T\bs{b} > 0$ (see \cref{lem:Farkas-primal}).
				
				\begin{lemma}\label{lem:primal-infeasible}
					If 
					\textsc{PrimalFeasibility} returns $\false$, 
					then the primal LP problem~\eqref{eq:LP} is infeasible and \textsc{FarkasCertificateOfPrimalInfeasibility} returns a Farkas' certificate of the primal infeasibility.
				\end{lemma}
				\begin{proof}
					It suffices to show that the output $\bs{\dray}^*$ satisfies that $(\bs{\dray}^*)^TA \le \bs{0}$ and $(\bs{\dray}^*)^T\bs{b} > 0$.
					
					We first show that $(\bs{\dray}^*)^TA \le \bs{0}$.
					It suffices to show that $\dray^*(h(E)) - \sum_{u \in T(E)}\weight(E,u)\dray^*(u) \le 0$ for each $E \in \mathcal{E}$, where we define $\dray^*(\emptyset) = 0$.
					
					\paragraph{Case 1: $h(E) = v$ for some $v \in V$ and $\nontriv^{(m)}(v) = \true$.}
					In this case, $\dray^*(h(E)) = 0$ by the definition of $\bs{\dray}^*$.
					Moreover, from \cref{cl:nontriv=constant} the coefficient of $M$ in $\dual^{(k)}(v)$ is nonnegative for any $k = 0,\dots, m$ and $v \in V$.
					Hence we have $\bs{\dray}^* \ge 0$ and this, together with $\weight(E,u) > 0$ for every $u \in T(E)$, implies that $- \sum_{u \in T(E)}\weight(E,u)\dray^*(u) \le 0$.
					
					\paragraph{Case 2: $h(E) = v$ for some $v \in V$ and $\nontriv^{(m)}(v) = \false$.}
					Since $\dray^*(v) = $ the coefficient of $M$ in $\dual^{(m)}(v)$, 
					if the coefficient of $M$ in $\dual^{(m)}(v)$ is at most the coefficient of $M$ in $\ell(E) + \sum_{u \in T(E)}\weight(E,u)\dual^{(m)}(u)$, then $\dray^*(v) - \sum_{u \in T(E)}\weight(E,u)\dray^*(u) \le 0$.
					
					Assume otherwise that the coefficient of $M$ in $\dual^{(m)}(v)$ is greater than the coefficient of $M$ in $\ell(E) + \sum_{u \in T(E)}\weight(E,u)\dual^{(m)}(u)$.
					We show that this case never occurs from gainfreeness of the LP problem.
					We first show that there exists $u \in T(E)$ such that $\nontriv^{(m)}(u)=\nontriv^{(m-1)}(u)=\false$, $\change^{(m)}(u)=\true$, and the coefficient of $M$ in $\dual^{(m)}(u)$ is less than the coefficient of $M$ in $\dual^{(m-1)}(u)$.
					To show this, first observe that from $\dual^{(m)}(v) \le \ell(E)+\sum_{u \in T(E)}\weight(E,u)\dual^{(m-1)}(u)$ we have 
					\begin{align}\label{eq:primal-infeasible-1}
						\begin{split}
							{\rm the\ coefficient\ of\ }M {\rm \ in\ }\dual^{(m)}(v) &= {\rm the\ coefficient\ of\ }M {\rm \ in\ }\sum_{u \in T(p^{(m)}(v))}\weight(p^{(m)}(v),u)\dual^{(m-1)}(u)\\
							&\le {\rm the\ coefficient\ of\ }M {\rm \ in\ }\sum_{u \in T(E)}\weight(E,u)\dual^{(m-1)}(u)\\
							&= {\rm the\ coefficient\ of\ }M {\rm \ in\ }\sum_{u \in T(E)\setminus V^{(m-1)}_{\nontriv}}\weight(E,u)\dual^{(m-1)}(u),
						\end{split}
					\end{align}
					where we note that the coefficient of $\dual^{(m-1)}(u)$ is zero for $u \in V^{(m-1)}_{\nontriv}$ by \cref{cl:nontriv=constant}.
					On the other hand, since the coefficient of $M$ in $\dual^{(m)}(v)$ is greater than the coefficient of $M$ in $\ell(E) + \sum_{u \in T(E)}\weight(E,u)\dual^{(m)}(u)$, we have 
					\begin{align}\label{eq:primal-infeasible-2}
						\begin{split}
							{\rm the\ coefficient\ of\ }M {\rm \ in\ }\dual^{(m)}(v) &> {\rm the\ coefficient\ of\ }M {\rm \ in\ }\sum_{u \in T(E)}\weight(E,u)\dual^{(m)}(u)\\
							&= {\rm the\ coefficient\ of\ }M {\rm \ in\ }\sum_{u \in T(E)\setminus V^{(m)}_{\nontriv}}\weight(E,u)\dual^{(m)}(u)\\
							&= {\rm the\ coefficient\ of\ }M {\rm \ in\ }\sum_{u \in T(E)\setminus V^{(m-1)}_{\nontriv}}\weight(E,u)\dual^{(m)}(u),
						\end{split}
					\end{align}
					where we have $V^{(m)}_{\nontriv} = V^{(m-1)}_{\nontriv}$ from $|V^{(m)}_{\nontriv}| \le |V \setminus \{v\}| = m-1$ and \cref{cl:nontrivial-level}.
					From \cref{eq:primal-infeasible-1,eq:primal-infeasible-2}, we obtain 
					\begin{align*}
						\begin{split}
							&{\rm the\ coefficient\ of\ }M {\rm \ in\ }\sum_{u \in T(E)\setminus V^{(m-1)}_{\nontriv}}\weight(E,u)\dual^{(m)}(u) \\
							&< {\rm the\ coefficient\ of\ }M {\rm \ in\ }\sum_{u \in T(E)\setminus V^{(m-1)}_{\nontriv}}\weight(E,u)\dual^{(m-1)}(u).
						\end{split}
					\end{align*}
					Hence, $T(E)\setminus V^{(m-1)}_{\nontriv} \neq \emptyset$ and we can choose $u_{m} \in T(E)$ such that $\nontriv^{(m)}(u_{m})=\nontriv^{(m-1)}(u_{m})=\false$, $\change^{(m)}(u_{m})=\true$, and the coefficient of $M$ in $\dual^{(m)}(u_{m})$ is less than the coefficient of $M$ in $\dual^{(m-1)}(u_{m})$.
					Let $E_k:=p^{(k)}(u_{k})$ for $k=1,...,m$.
					Now, we inductively show that for $k=m,\dots,2$ 
					we can choose $u_{k-1} \in T(E_k)$ such that $\nontriv^{(k-1)}(u_{k-1})=\nontriv^{(k-2)}(u_{k-1})=\false$, $\change^{(k-1)}(u_{k-1})=\true$, and the coefficient of $M$ in $\dual^{(k-1)}(u_{k-1})$ is less than the coefficient of $M$ in $\dual^{(k-2)}(u_{k-1})$ until $u_k = u_{\ell}$ for some $\ell > k$.
					
					For $k=m$, from $|V^{(m-1)}_{\nontriv}| \le |V \setminus \{v,u_m\}| = m-2$ we have $V^{(m-1)}_{\nontriv} = V^{(m-2)}_{\nontriv}$ (\cref{cl:nontrivial-level}).
					Hence, together with $\dual^{(m)}(u_m) = \ell(E_m)+\sum_{u \in T(E_m)}\weight(E,u)\dual^{(m-1)}(u)$, we have 
					\begin{align}\label{eq:primal-infeasible-3}
						\begin{split}
							{\rm the\ coefficient\ of\ }M {\rm \ in\ }\dual^{(m)}(u_m) &= {\rm the\ coefficient\ of\ }M {\rm \ in\ }\sum_{u \in T(E_m)}\weight(E,u)\dual^{(m-1)}(u)\\
							&= {\rm the\ coefficient\ of\ }M {\rm \ in\ }\sum_{u \in T(E_m)\setminus V^{(m-1)}_{\nontriv}}\weight(E,u)\dual^{(m-1)}(u)\\
							&= {\rm the\ coefficient\ of\ }M {\rm \ in\ }\sum_{u \in T(E_m)\setminus V^{(m-2)}_{\nontriv}}\weight(E,u)\dual^{(m-1)}(u).
						\end{split}
					\end{align}
					On the other hand, since the coefficient of $M$ in $\dual^{(m)}(u_m)$ is less than the coefficient of $M$ in $\dual^{(m-1)}(u_m)$, we have \begin{align}\label{eq:primal-infeasible-4}
						\begin{split}
							{\rm the\ coefficient\ of\ }M {\rm \ in\ }\dual^{(m)}(u_m) &< {\rm the\ coefficient\ of\ }M {\rm \ in\ }\dual^{(m-1)}(u_m)\\
							&= {\rm the\ coefficient\ of\ }M {\rm \ in\ }\sum_{u \in T(E_{m-1})}\weight(E,u)\dual^{(m-2)}(u)\\
							&\le {\rm the\ coefficient\ of\ }M {\rm \ in\ }\sum_{u \in T(E_m)}\weight(E,u)\dual^{(m-2)}(u)\\
							&= {\rm the\ coefficient\ of\ }M {\rm \ in\ }\sum_{u \in T(E_m)\setminus V^{(m-2)}_{\nontriv}}\weight(E,u)\dual^{(m-2)}(u).
						\end{split}
					\end{align}
					From \cref{eq:primal-infeasible-3,eq:primal-infeasible-4}, we obtain 
					\begin{align*}
						\begin{split}
							&{\rm the\ coefficient\ of\ }M {\rm \ in\ }\sum_{u \in T(E_m)\setminus V^{(m-2)}_{\nontriv}}\weight(E,u)\dual^{(m-1)}(u) \\
							&< {\rm the\ coefficient\ of\ }M {\rm \ in\ }\sum_{u \in T(E_m)\setminus V^{(m-2)}_{\nontriv}}\weight(E,u)\dual^{(m-2)}(u).
						\end{split}
					\end{align*}
					Hence, $T(E_m)\setminus V^{(m-2)}_{\nontriv} \neq \emptyset$ and we can choose $u_{m-1} \in T(E_m)$ such that $\nontriv^{(m-1)}(u_{m-1})=\nontriv^{(m-2)}(u_{m-1})=\false$, $\change^{(m-1)}(u_{m-1})=\true$, and the coefficient of $M$ in $\dual^{(m-1)}(u_{m-1})$ is less than the coefficient of $M$ in $\dual^{(m-2)}(u_{m-1})$.
					
					For $k<m$, assume that $u_{m+1}, u_m, \dots, u_k$ are distinct.
					Since $\nontriv^{(k-1)}(u_{\ell})=\false$ for each $k \le \ell \le m$, we have $|V^{(k-1)}_{\nontriv}| \le k-2$ and thus $V^{(k-1)}_{\nontriv} = V^{(k-2)}_{\nontriv}$ (\cref{cl:nontrivial-level}).
					Then we can similarly show that there exists $u_{k-1} \in E_k$ such that $\nontriv^{(k-1)}(u_{k-1})=\nontriv^{(k-2)}(u_{k-1})=\false$, $\change^{(k-1)}(u_{k-1})=\true$, and the coefficient of $M$ in $\dual^{(k-1)}(u_{k-1})$ is less than the coefficient of $M$ in $\dual^{(k-2)}(u_{k-1})$.
					
					Since there exists $m$ vertices in the graph, for some $s < t \in \{1, \dots, m+1\}$ we have 
					$u_s = u_t$, where we define $u_{m+1}:=v$ and $E_{m+1}:=E$.
					Choose such $s,t$ where $t$ is maximum and the difference $|t-s|$ is minimum.
					Then, $u_t, E_t, u_{t-1}, E_{t-1}, \dots, E_{s+1},u_s$ is a cycle in $\mathcal{H}$.
					Moreover, for each $k=2,\dots, m+1$ $\dual^{(k)}(u_k) = \ell(E_k) + \sum_{u \in T(E_k)}\weight(E_k,u)\dual^{(k-1)}(u)$ implies that the coefficient of $M$ in $\dual^{(k)}(u_k)$ is at least the coefficient of $M$ in $\weight(E_k,u_{k-1})\dual^{(k-1)}(u_{k-1})$, 
					where we define $\dual^{(m+1)}(u_{m+1}) := \ell(E_{m+1}) + \sum_{u \in T(E_{m+1})}\weight(E_{m+1},u)\dual^{(m)}(u)$.
					Hence,  we have 
					\begin{align*}
						{\rm the\ coefficient\ of\ }M {\rm \ in\ }\dual^{(t)}(u_{t}) &\ge {\rm the\ coefficient\ of\ }M {\rm \ in\ }\weight(E_{t},u_{t-1})\dual^{(t-1)}(u_{t-1})\\
						&\ge {\rm the\ coefficient\ of\ }M {\rm \ in\ }\weight(E_{t},u_{t-1})\weight(E_{t-1},u_{t-2})\dual^{(t-2)}(u_{t-2})\\
						&\ge \dots \\
						&\ge {\rm the\ coefficient\ of\ }M {\rm \ in\ }\prod_{k=s+1}^t \weight(E^{(k)},u_{k-1}) \dual^{(s)}(u_{s})\\
						&= {\rm the\ coefficient\ of\ }M {\rm \ in\ }\prod_{k=s+1}^t \weight(E^{(k)},u_{k-1}) \dual^{(s)}(u_{t}).
					\end{align*}
					Since ${\rm the\ coefficient\ of\ }M {\rm \ in\ }\dual^{(t)}(u_{t}) < {\rm the\ coefficient\ of\ }M {\rm \ in\ }\dual^{(s)}(u_{t})$ (as $\nontriv^{(t)}(u_{t}) = \false$, $\change^{(t)}(u_{t})=\true$, and the coefficient of $M$ in $\dual^{(t)}(u_t)$ is less than the coefficient of $M$ in $\dual^{(t-1)}(u_{t})$ and $\dual^{(k)}(u_{t})$ is monotone-decreasing in $k$), we obtain that $\prod_{k=s+1}^t \weight(E^{(k)},u_{k-1}) \le \frac{\dual^{(t)}(u_{t})}{\dual^{(s)}(u_{t})} < 1$.
					This contradicts that the system is gainfree.
					Hence, the coefficient of $M$ in $\dual^{(m)}(v)$ is at most the coefficient of $M$ in $\ell(E) + \sum_{u \in T(E)}\weight(E,u)\dual^{(m)}(u)$.

					\paragraph{Case 3: $h(E) = \emptyset$.}
					This case can be shown in a similar way as in Case 1.
					This completes the proof of $(\bs{\dray}^*)^TA \le \bs{0}$.
					
					We then show that $(\bs{\dray}^*)^T\bs{b} > 0$.
					Note that we have $\bs{\dray}^* \ge \bs{0}$ from the observation in Case 1 above.
					Moreover, from \cref{cl:nontriv=constant} the coefficient of $M$ in $\dual^{(k)}(v)$ is positive for any $k = 0,\dots, m$ and $v \in V$ with $\nontriv^{(k)}(v) = \false$, 
					implying that $\dray^*(v) > 0$ for every $v \in V$ with $\nontriv^{(m)}(v) = \false$.
					Also, $\bs{b} \ge \bs{0}$ by definition.
					Since $b(v) > 0$ for some $v \in V$ with $\nontriv^{(m)}(v) = \false$ and $\dray^*(v) > 0$ for such $v$, 
					we have $(\bs{\dray}^*)^T\bs{b} = \sum_{v \in V} \dray^*(v) b(v) > 0$.
					This completes the proof.
				\end{proof}
				
				Now, we are ready to show \cref{prop:primal-feasibility-correct}.
				
				\begin{proof}[Proof of \cref{prop:primal-feasibility-correct}]
					Note that subroutines \textsc{PrimalFeasibility}, \textsc{PrimalSolution}, and \textsc{FarkasCertificateOfPrimalInfeasibility}, together with \textsc{DualFeasibility}, constitute a certifying algorithm for the feasibility problem of the primal LP problem~\eqref{eq:LP} (\cref{alg:cert-gainfree-Leontief-primal}).
					The correctness of this algorithm follows from \cref{lem:primal-feasible,lem:primal-infeasible}. 
					
					Now, we analyze the running time of the above algorithm.
					The most time-consuming part of is \textsc{DualFeasibility}, which runs in ${\rm O}(m^3n)$ time as shown in the proof of \cref{prop:dual-feasibility-correct}.
					This completes the proof.
				\end{proof}
				
				\subsection{Proof of the main theorem (Theorem~\ref{thm:main-theorem})}
				\label{subsec:proof-of-main-theorem}
				
				Combining the results in \cref{subsec:dual-feasibility,subsec:primal-feasibility}, we obtain our main theorem, which we recall:
				
				\newtheorem*{MainThm}{Theorem~\ref{thm:main-theorem}}
				\begin{MainThm}[Main]
					The LP problems with gainfree Leontief substitution systems \eqref{eq:LP} admit a combinatorial 
					${\rm O}(m^3n)$-time 
					certifying algorithm. 
				\end{MainThm}
				
				\begin{proof}
					From \cref{thm:LP-solution-pattern}, 
					\cref{alg:cert-gainfree-Leontief-dual,alg:cert-gainfree-Leontief-primal} constitute a certifying algorithm for solving the LP problem.
					The correctness and the running time of the algorithm follow from \cref{prop:dual-feasibility-correct,prop:primal-feasibility-correct}.
				\end{proof}
				
				\paragraph{Gupta's algorithm}
				An attempt has been made to devise a certifying combinatorial algorithm for the feasibility of unit Horn systems with nonpositive variables~\cite{Gup14}.
				They try to devise an algorithm that outputs a feasible solution if a given system is feasible, and a Farkas' certificate if not.
				Their algorithm is based on the Bellman-Ford algorithm and a directed graph representation of a unit Horn system. 
				However, their statements of the algorithm are not precise and include some mistakes even in certain proofs of correctness of the algorithm.
				Regardless of how we fix their algorithm, 
				we can create instances in which their algorithm fails.
				For example, their algorithm in its current form fails for the instance in~\cref{ex:main-algo1}.
				Moreover, we remark that they mentioned that it is open to extend their claimed results on unit Horn systems with nonpositive variables to (i) unit Horn systems where variables are allowed to have positive values and (ii) unit-positive Horn systems. In this paper, we resolve both of these issues.
					
				\section{Discussions}
				\label{sec:discussions}
				In this section, we first consider integer versions of the primal and dual LP problems and provide analysis of computational complexity and algorithms in \cref{subsec:integer-version}.
				Then we address the question whether we can obtain a combinatorial certifying algorithm for LP problems with gainfree Leontief substitution systems by incorporating the idea of two-phase simplex method and using an existing non-certifying combinatorial algorithm for the problem in \cref{subsec:two-phase-does-not-work}.

				\subsection{Integer versions of the primal and the dual of LP problems with Gainfree Leontief substitution systems}
				\label{subsec:integer-version}
				Here, 
				we consider solving the following integer programming problems:
				
				\begin{align}
					\label{eq:IP}
					\begin{array}{ll}
						\rm{minimize} & \bs{c}^T\bs{\primal}\\
						\rm{subject\ to} & A\bs{\primal} = \bs{b}\\
						& \bs{\primal} \in \mathbb{Z}^n_+,
					\end{array}
				\end{align}
				and
				\begin{align}
					\label{eq:dual-IP}
					\begin{array}{ll}
						\rm{maximize} & \bs{\dual}^T\bs{b}\\
						\rm{subject\ to} & \bs{\dual}^TA \leq \bs{c}^T\\
						& \bs{\dual} \in \mathbb{Z}^m,
					\end{array}
				\end{align}
				where $A$ is gainfree Leontief and $b \ge \bs{0}$, and $\mathbb{Z}$ and $\mathbb{Z}_{+}$ denote the sets of integers and nonnegative integers, respectively.
				
				To consider integrality of solutions and certificates, we assume that the data given is all rational rather than real in this subsection.
				Note that our certifying algorithm in the previous section remains polynomial time in this setting, since the bit length of the data appearing in the algorithm can be polynomially bounded.
				Moreover, we can make Farkas' certificates of infeasibility of the primal and dual LP problems integer by multiplying an integer to rational Farkas certificates.
				
				We first consider the integer version of the dual of LP problems with gainfree Leontief substitution systems~\eqref{eq:dual-IP}, which is the generalization of the integer version of the LP problems with unit-positive Horn systems consider in the literature.
				Then we turn into the integer version of the primal LP problems
				with gainfree Leontief substitution systems~\eqref{eq:IP}.
				
				\subsubsection*{Integer version of the dual LP problem}
				Here, we consider the integer version of the dual LP problem with gainfree Leontief substitution systems~\eqref{eq:dual-IP}.
				Observe that the subroutine \textsc{DualSolution} outputs an integer vector when $A$ is an integer matrix and $\bs{c}$ is an integer vector.
				In this case, the constraint of the dual LP problem is exactly the unit-positive Horn system and thus, it has an integer optimal solution when it has an optimal solution~\cite{JMR92,ChS13}.
				Hence, our algorithm is also a combinatorial certifying algorithm for the integer feasibility of the unit-positive Horn systems.
				Recall that this resolves the open problems raised in\cite{Gup14}.
				When $\bs{c}$ might not be an integer vector, we can solve the problem by replacing $\bs{c}$ with $\lfloor\bs{c}\rfloor$.
				
				When $A$ might not be an integer matrix, the situation changes.
				In fact, we point out that it is NP-complete to determines the integer feasibility of the dual LP problem with gainfree Leontief substitution systems.
				This is because an NP-complete problem is actually reduced to the integer feasibility of the dual LP problem with gainfree Leontief substitution systems in the NP-completeness proof of the integer feasibility of the Horn systems in~\cite{Lag85}.
				\begin{theorem}[Follows from \cite{Lag85}]
					The integer feasibility of the dual LP problem~\eqref{eq:dual-IP} with gainfree Leontief substitution systems is NP-complete.
				\end{theorem}
				
				Moreover, we may consider the integer version of \textsc{DualFeasibility} for the feasibility problem of the dual LP problem, 
				in which values $\bs{\dual}^{(k)}$ are updated by taking floor function, i.e., 
				the statement in line 6 is replaced by ``$\dual^{(k)}(v)\ot \lfloor \ell(E) + \sum_{u \in T(E)}\weight(E,u)\dual^{(k-1)}(u) \rfloor$.''
				However, this algorithm requires exponential time in the worst case; see \cref{ex:exponential-integer-feasibility-gainfree} below.
				
				\begin{example}
					\label{ex:exponential-integer-feasibility-gainfree}
					Let $a$ be an positive integer.
					Consider the integer feasibility of the following dual LP problem with a gainfree Leontief substitution system: 
					\begin{align}
						\label{eq:exponential-integer-feasibility-gainfree}
						\left\{
						\begin{array}{rll}
							-\frac{a+1}{a}y_1 + y_2 &\le& 1\\
							y_1 - \frac{a}{a+1}y_2 &\le& - \frac{a}{a+1}\\
							y_1 &\le& 0\\
							y_2 &\le& 0\\
							y_1,y_2 &\in& \mathbb{Z}.
						\end{array}
						\right.
					\end{align}
					It can be seen that in the for-loop from line 2 to 19 of the integer version of \textsc{DualFeasibility} mentioned above, 
					the dual variables are updated as 
					$\dual^{(k)}(v_1) = -\lfloor k/2 \rfloor$ and 
					$\dual^{(k)}(v_2) = -\lfloor (k-1)/2 \rfloor$
					for $k=1,...,2a+1$.
					Hence, the procedure takes ${\rm \Omega}(a)$ time to obtain a feasible solution ($\bs{\dual}=(-a,-a)$) of \eqref{eq:exponential-integer-feasibility-gainfree}.
					This is exponential in the input size $\log (a)$ of the number $a$.
				\end{example}
				
				Conversely, by setting $M$ in the integer version of \textsc{DualFeasibility} sufficiently large exponential number $\eta:=2^{13m(\size(A)+\size(c)}$ and also iterate $2m\eta$ times the for-loop from line 2 to line 19, we obtain an exponential time algorithm for the feasibility of \eqref{eq:dual-IP}.
				Indeed, if \eqref{eq:dual-IP} is feasible, the algorithm finds a feasible solution since the size of such solution can be bounded by $\eta$ (see, e.g., Corollary 5.8 in~\cite{KoF18}).
				If the algorithm finds a feasible solution, then we solve the original \eqref{eq:dual-LP}.
				If \eqref{eq:dual-LP} is bounded, so is \eqref{eq:dual-IP} and thus, we obtain an optimal solution of \eqref{eq:dual-IP} without a certificate.
				If the input data are rational and \eqref{eq:dual-LP} is unbounded, it is known (e.g., Proposition 5.2 in \cite{KoF18}) that \eqref{eq:dual-IP} is also unbounded, and we can obtain a Farkas certificate of the primal infeasibility as a certificate of the unboundedness.
				If the integer version of \textsc{DualFeasibility} does not find a feasible solution, then \eqref{eq:dual-IP} is infeasible; however, we do not know if we can obtain a certificate of infeasibility, since the primal LP problem is not a dual problem of \eqref{eq:dual-IP}.
				
				We remark that \eqref{eq:dual-IP} can be solved as follows.
				First, execute the exponential time algorithm of integer version of \textsc{DualFeasibility}.
				If it outputs a feasible solution $\bs{\dual}^*$, then set $M$ to $\eta+\alpha$ and also iterate $2m\alpha$ times the for-loop from line 2 to line 19, where $\alpha:=m\Xi(A)$ and $\Xi(A)$ is the maximum absolute value of the subdeterminants of $A$.
				From Theorem5.7 in~\cite{KoF18}, if $\bs{\dual}^*$ is not an optimal solution, then there exists a feasible integer solution $\bs{\duall}^*$ with $\bs{b}^T\bs{\duall}^*>\bs{b}^T\bs{\dual}^*$.
				Hence, if we obtain a feasible solution that has lager objective value than 
				$\bs{\dual}^*$, then we can conclude that the problem is unbounded; otherwise, $\bs{\dual}^*$ is an optimal solution of \eqref{eq:dual-IP}.
				
				An open problem is to settle whether there exists a pseudo-polynomial time algorithm to solve the integer version of the dual LP problems with gainfree Leontief substitution systems.
				
				\subsubsection*{Integer version of the primal LP problem}
				Here, we consider the integer version of the primal LP problem with gainfree Leontief substitution systems~\eqref{eq:IP}.
				We first consider the case where $A$ is an integer matrix, i.e., $A$ is a unit-positive integer matrix.
				In this case, if $\bs{b}$ is not an integer vector, then \eqref{eq:IP} is infeasible, since $A\bs{\primal}$ is an integer vector for any $\bs{\primal} \in \mathbb{Z}_{+}^n$.
				Assume that $\bs{b}$ is an integer vector.
				Then, from Theorem 4.1 in~\cite{JMR92}, the set of feasible solutions of \eqref{eq:IP} is an integer polyhedra.
				This implies that \eqref{eq:IP} can be solved by an algorithm for \eqref{eq:LP}.
				Hence, our algorithm for \eqref{eq:LP} is a combinatorial certifying algorithm for \eqref{eq:IP}.
				
				When $A$ might not be an integer matrix, we show that it is NP-complete to determines the integer feasibility of the primal LP problem with gainfree Leontief substitution systems.
				
				\begin{theorem}
					The integer feasibility of the primal LP problem~\eqref{eq:IP} with gainfree Leontief substitution systems is NP-complete.
				\end{theorem}
				\begin{proof}
					This follows from the fact that the unbounded subset sum problem, which is NP-complete (see, e.g., \cite{HR96}), can be formulated as the primal LP problem~\eqref{eq:IP} with gainfree Leontief substitution systems.
					The unbounded subset sum problem is, given positive integers $a_1,\dots, a_n$ and $b$, to determine if there exists nonnegative integers $x_1,\dots, x_n$ satisfying $a_1x_1+\dots +a_nx_n =b$.
					Hence, the problem has the form of \eqref{eq:IP}.
					Moreover, the corresponding hypergraph contains no directed cycle, and thus the system is gainfree.
				\end{proof}

				Since the dual LP problem~\eqref{eq:dual-LP} is not a dual problem of the integer version of the primal LP problem with gainfree Leontief substitution systems~\eqref{eq:IP}, it seems difficult to apply our certifying algorithm to \eqref{eq:IP}.
				
				\subsection{Two-phase method}
				\label{subsec:two-phase-does-not-work}
				
				As mentioned in \cref{sec:introduction}, 
				in the two-phase simplex method one transforms the feasibility of an LP problem into an LP problem which always has an optimal solution.
				Hence, it seems that we can use a combinatorial algorithm for the LP problem only certifying for the case where there exists an optimal solution.
				We show that this idea actually works for the primal LP problem~\eqref{eq:LP}, while not for the dual LP problem~\eqref{eq:dual-LP} in what follows.
				
				\subsubsection*{Primal LP problems}
				For the sake of clarity, we rewrite the feasibility of the primal LP problem~\eqref{eq:LP} as the feasibility of the following system: 
				\begin{align}
					\label{eq:linear-system}
					\left\{
					\begin{aligned}
						A\bs{\primal} = \bs{b},\\
						\bs{\primal} \ge \bs{0}.    
					\end{aligned}
					\right.
				\end{align}
				
				The auxiliary LP problem that determines the feasibility of \eqref{eq:linear-system} is as follows: 
				
				\begin{eqnarray}\label{eq:two-phase-LP}
					\begin{array}{cl}
						\rm{minimize}     &\displaystyle \sum_{i=1}^{m}s_i+\sum_{i=1}^{m}t_i\\
						\rm{subject\ to}     &A\bs{\primal}+\bs{s}-\bs{t} = \bs{b}\\
						&\bs{x},\bs{s},\bs{t} \ge \bs{0}.
					\end{array}
				\end{eqnarray}
				The dual LP problem of \eqref{eq:two-phase-LP} is 
				\begin{eqnarray}\label{eq:dual-of-two-phase-LP}
					\begin{array}{cl}
						\rm{maximize}     &\bs{b}^T\bs{\dual}\\
						\rm{subject\ to}     &A^T\bs{\dual} \le \bs{0}\\
						&\ \ \bs{\dual} \le \bs{1}\\
						&-\bs{\dual} \le \bs{1}.
					\end{array}
				\end{eqnarray}
				
				Note that \eqref{eq:linear-system} is feasible if and only if the optimal value of \eqref{eq:two-phase-LP} is zero, and that \eqref{eq:two-phase-LP} has an optimal solution since it is bounded below.
				Since the constraint matrix in \eqref{eq:two-phase-LP} is again gainfree Leontief, the combinatorial algorithm in~\cite{JMR92} solves \eqref{eq:two-phase-LP} and outputs optimal solutions $(\bs{\primal}^*,\bs{s}^*,\bs{t}^*)$ and $\bs{\dual}^*$ of \eqref{eq:two-phase-LP} and its dual \eqref{eq:dual-of-two-phase-LP}, respectively.
				If \eqref{eq:linear-system} is feasible, then $\bs{s}^*=\bs{t}^*=\bs{0}$ and thus, $\bs{\primal}^*$ is a feasible solution of \eqref{eq:linear-system}. 
				If \eqref{eq:linear-system} is infeasible, then $\bs{b}^T\bs{\dual}^*>0$ and $A^T\bs{\dual}\le \bs{0}$, since the optimal value is greater than zero.
				Hence, $\bs{\dual}^*$ is a Farkas' certificate of the infeasibility of \eqref{eq:linear-system}.
				Therefore, we obtain a combinatorial certifying algorithm for the feasibility of the primal LP problem~\eqref{eq:LP}.
				In our combinatorial certifying algorithm in the previous section, the symbol $M$ enables us to directly compute a Farkas' certificate of the infeasibility of \eqref{eq:linear-system}.
				
				\subsubsection*{Dual LP problems}
								For the sake of clarity, we rewrite the feasibility of the dual of the LP problem with a gainfree Leontief substitution system as the feasibility of the following system: 
				\begin{align}
					\label{eq:dual-linear-system-transposed}
					\begin{array}{l}
						A^T\bs{\dual} \le \bs{c}.\\
										\end{array}
									\end{align}
				
				We can determine the feasibility of the system \eqref{eq:dual-linear-system-transposed} by solving an LP problem with auxiliary variables as follows.
				
				\begin{eqnarray}\label{eq:two-phase-dual-LP}
					\begin{array}{cl}
						\rm{maxmize}     &\displaystyle \sum_{i=1}^{m}t_i\\
						\rm{subject\ to}     &A^T\bs{\dual}+\bs{t} \le \bs{c}\\
						&\bs{t} \le \bs{0}\\
						&\bs{\dual} \in \mathbb{R}^m.
					\end{array}
				\end{eqnarray}
				
				Note that \eqref{eq:dual-linear-system-transposed} is feasible if and only if the optimal value of \eqref{eq:two-phase-dual-LP} is zero, and that \eqref{eq:two-phase-LP} has an optimal solution since it is bounded above.
				However, the constraint of the LP problem~\eqref{eq:two-phase-dual-LP} is no more Leontief, and thus combinatorial algorithms for the gainfree Leontief substitution systems cannot be applied to the LP problem~\eqref{eq:two-phase-dual-LP}.
				
				\section{Conclusion}
				\label{sec:conclusion}
				
				We proposed a certifying algorithm for the LP problems with gainfree Leontief substitution systems.
				Our algorithm is combinatorial and runs in strongly polynomial time.
				Since the dual LP problems with gainfree Leontief substitution systems contains the feasibility of unit-positive Horn systems, 
				we resolved the open questions raised in \cite{Gup14}.
				
				An interesting future direction would be to make other non-certifying algorithms certifying.
				A candidate would be to extend our result on unit Horn systems to unit \emph{q-Horn} systems, introduced in~\cite{KiM16}.
				Unit q-Horn systems include not only unit Horn systems but also unit-two-variable-per-inequality (UTVPI) systems, and 
				the feasibility problem of unit q-Horn systems is solvable in polynomial time~\cite{KiM16}.
				Furthermore, a certifying algorithm for the feasibility problem of UTVPI systems is known~\cite{LaM05}. 
				Therefore, giving a certifying algorithm to the feasibility of unit q-Horn systems would be an interesting future work.

				\bibliography{Horn}

\begin{thebibliography}{10}

\bibitem{AdC91}
Ilan Adler and Steven Cosares.
\newblock A strongly polynomial algorithm for a special class of linear
  programs.
\newblock {\em Operations Research}, 39:955--960, 1991.

\bibitem{Bel58}
Richard Bellman.
\newblock On a routing problem.
\newblock {\em Quarterly of applied mathematics}, 16(1):87--90, 1958.

\bibitem{BGM22}
Bart Bogaerts, Stephan Gocht, Ciaran McCreesh, and Jakob Nordstr{\"o}m.
\newblock Certified symmetry and dominance breaking for combinatorial
  optimisation.
\newblock In {\em Proceedings of the 36th AAAI Conference on Artificial
  Intelligence (AAAI'22)}, 2022.

\bibitem{CGS97}
Riccardo Cambini, Giorgio Gallo, and Maria~Grazia Scutell{\`a}.
\newblock Flows on hypergraphs.
\newblock {\em Mathematical Programming}, 78(2):195--217, 1997.

\bibitem{ChS13}
R.~Chandrasekaran and K.~Subramani.
\newblock A combinatorial algorithm for horn programs.
\newblock {\em Discrete Optimization}, 10:85--101, 2013.

\bibitem{CGS17}
Kevin~KH Cheung, Ambros Gleixner, and Daniel~E Steffy.
\newblock Verifying integer programming results.
\newblock In {\em International Conference on Integer Programming and
  Combinatorial Optimization}, pages 148--160. Springer, 2017.

\bibitem{CGS16}
Maria Chudnovsky, Jan Goedgebeur, Oliver Schaudt, and Mingxian Zhong.
\newblock Obstructions for three-coloring graphs with one forbidden induced
  subgraph.
\newblock In {\em Proceedings of the twenty-seventh annual ACM-SIAM symposium
  on Discrete algorithms}, pages 1774--1783. SIAM, 2016.

\bibitem{CDH13}
Derek~G Corneil, Barnaby Dalton, and Michel Habib.
\newblock Ldfs-based certifying algorithm for the minimum path cover problem on
  cocomparability graphs.
\newblock {\em SIAM Journal on Computing}, 42(3):792--807, 2013.

\bibitem{CoA72}
Richard~W. Cottle and Arthur~F. {Veinott, Jr.}
\newblock Polyhedral sets having a least element.
\newblock {\em Mathematical Programming}, 3:238--249, 1972.

\bibitem{Dan55}
George~B. Dantzig.
\newblock Optimal solution of a dynamic leontief model with substitution.
\newblock {\em Econometrica}, 23(3):295--302, 1955.

\bibitem{DFK03}
Marcel Dhiflaoui, Stefan Funke, Carsten Kwappik, Kurt Mehlhorn, Michael Seel,
  Elmar Sch{\"o}mer, Ralph Schulte, and Dennis Weber.
\newblock Certifying and repairing solutions to large lps how good are
  lp-solvers?
\newblock In {\em Proceedings of the fourteenth annual ACM-SIAM symposium on
  Discrete algorithms}, pages 255--256, 2003.

\bibitem{Tar86}
\'{E}va Tardos.
\newblock A strongly polynomial algorithm to solve combinatorial linear
  programs.
\newblock {\em Operations Research}, 34:250--256, 1986.

\bibitem{For56}
Lester~R Ford~Jr.
\newblock Network flow theory.
\newblock Technical report, Rand Corp Santa Monica Ca, 1956.

\bibitem{GeT15}
Loukas Georgiadis and Robert~E Tarjan.
\newblock Dominator tree certification and divergent spanning trees.
\newblock {\em ACM Transactions on Algorithms (TALG)}, 12(1):1--42, 2015.

\bibitem{Glo64}
Fred Glover.
\newblock A bound escalation method for the solution of integer linear
  programs.
\newblock {\em Cahiers du Centre d'Etudes de Recherche Operationelle},
  6(3):131--168, 1964.

\bibitem{Gol95}
Andrew~V Goldberg.
\newblock Scaling algorithms for the shortest paths problem.
\newblock {\em SIAM Journal on Computing}, 24(3):494--504, 1995.

\bibitem{Gup14}
Pratik~Bijaiprakash Gupta.
\newblock A certifying algorithm for {H}orn constraint systems.
\newblock Master's thesis, The University of Texas at Dallas, 2014.

\bibitem{HR96}
Paul Hansen and Jennifer Ryan.
\newblock Testing integer knapsacks for feasibility.
\newblock {\em European Journal of Operational Research}, 88:578--582, 1996.

\bibitem{JMR92}
Robert~G Jeroslow, Kipp Martin, Ronald~L Rardin, and Jinchang Wang.
\newblock Gainfree leontief substitution flow problems.
\newblock {\em Mathematical Programming}, 57(1):375--414, 1992.

\bibitem{KaN09}
Haim Kaplan and Yahav Nussbaum.
\newblock Certifying algorithms for recognizing proper circular-arc graphs and
  unit circular-arc graphs.
\newblock {\em Discrete Applied Mathematics}, 157(15):3216--3230, 2009.

\bibitem{KiM16}
Kei Kimura and Kazuhisa Makino.
\newblock Trichotomy for integer linear systems based on their sign patterns.
\newblock {\em Discrete Applied Mathematics}, 200:67--78, 2016.

\bibitem{KWS19}
Hans Kleine~B{\"u}ning, Piotr Wojciechowski, and K~Subramani.
\newblock New results on cutting plane proofs for {H}orn constraint systems.
\newblock In {\em 39th IARCS Annual Conference on Foundations of Software
  Technology and Theoretical Computer Science (FSTTCS 2019)}. Schloss
  Dagstuhl-Leibniz-Zentrum fuer Informatik, 2019.

\bibitem{KoF18}
Bernhard Korte and Jens Vygen.
\newblock {\em Combinatorial Optimization: Theory and Algorithms}.
\newblock Springer, sixth edition, 2018.

\bibitem{KMM06}
Dieter Kratsch, Ross~M McConnell, Kurt Mehlhorn, and Jeremy~P Spinrad.
\newblock Certifying algorithms for recognizing interval graphs and permutation
  graphs.
\newblock {\em SIAM Journal on Computing}, 36(2):326--353, 2006.

\bibitem{Lag85}
Jeffrey~C Lagarias.
\newblock The computational complexity of simultaneous diophantine
  approximation problems.
\newblock {\em SIAM Journal on Computing}, 14(1):196--209, 1985.

\bibitem{LaM05}
Shuvendu~K Lahiri and Madanlal Musuvathi.
\newblock An efficient decision procedure for utvpi constraints.
\newblock In {\em International Workshop on Frontiers of Combining Systems},
  pages 168--183. Springer, 2005.

\bibitem{MMN11}
Ross~M McConnell, Kurt Mehlhorn, Stefan N{\"a}her, and Pascal Schweitzer.
\newblock Certifying algorithms.
\newblock {\em Computer Science Review}, 5(2):119--161, 2011.

\bibitem{Meg83}
Nimrod Megiddo.
\newblock Towards a genuinely polynomial algorithm for linear programming.
\newblock {\em SIAM Journal on Computing}, 12:347--353, 1983.

\bibitem{MNN99}
Kurt Mehlhorn, Stefan Naher, and Stefan N{\"a}her.
\newblock {\em LEDA: A platform for combinatorial and geometric computing}.
\newblock Cambridge university press, 1999.

\bibitem{MNS17}
Kurt Mehlhorn, Adrian Neumann, and Jens~M Schmidt.
\newblock Certifying 3-edge-connectivity.
\newblock {\em Algorithmica}, 77(2):309--335, 2017.

\bibitem{Min06}
Antoine Min{\'e}.
\newblock The octagon abstract domain.
\newblock {\em Higher-order and symbolic computation}, 19(1):31--100, 2006.

\bibitem{Moo59}
Edward~F Moore.
\newblock The shortest path through a maze.
\newblock In {\em Proc. Int. Symp. Switching Theory, 1959}, pages 285--292,
  1959.

\bibitem{OlV20}
Neil Olver and L\'{a}szl\'{o}~A. V\'{e}gh.
\newblock A simpler and faster strongly polynomial algorithm for generalized
  flow maximization.
\newblock {\em Journal of the ACM}, 67:1--26, 2020.

\bibitem{Sch13}
Jens~M Schmidt.
\newblock Contractions, removals, and certifying 3-connectivity in linear time.
\newblock {\em SIAM Journal on Computing}, 42(2):494--535, 2013.

\bibitem{Sch98}
Alexander Schrijver.
\newblock {\em Theory of linear and integer programming}.
\newblock John Wiley \& Sons, 1998.

\bibitem{SuW17}
K~Subramani and Piotr Wojciechowski.
\newblock A combinatorial certifying algorithm for linear feasibility in utvpi
  constraints.
\newblock {\em Algorithmica}, 78(1):166--208, 2017.

\bibitem{SuW11}
K~Subramani and James Worthington.
\newblock A new algorithm for linear and integer feasibility in horn
  constraints.
\newblock In {\em International Conference on AI and OR Techniques in
  Constriant Programming for Combinatorial Optimization Problems}, pages
  215--229. Springer, 2011.

\bibitem{UvG88}
J.D. Ullman and A.~Van Gelder.
\newblock Efficient test for top-down termination of logical rules.
\newblock {\em Journal of the Association for Computing Machinery},
  35:345--373, 1988.

\bibitem{MaD02}
Hans Van~Maaren and Chuangyin Dang.
\newblock Simplicial pivoting algorithms for a tractable class of integer
  programs.
\newblock {\em Journal of Combinatorial Optimization}, 6(2):133--142, 2002.

\bibitem{Van20}
Robert~J. Vanderbei.
\newblock {\em Linear Programming: Foundations and Extensions}.
\newblock Springer, fifth edition, 2020.

\bibitem{Leo51}
Leontief Wassily~W.
\newblock {\em The structure of American economy, 1919-1939}.
\newblock Oxford University Press, second edition, 1951.

\end{thebibliography}
				\bibliographystyle{plain}

			\end{document}